\documentclass[letterpaper, 10 pt, conference]{ieeeconf}  

\IEEEoverridecommandlockouts                              

\usepackage{amsmath} 
\usepackage{amssymb}  
\usepackage{color}
\usepackage{graphics} 
\usepackage{graphicx}
\usepackage{epsfig}   
\usepackage{epstopdf}

\usepackage{amsthm}

\usepackage[noadjust]{cite} 

\usepackage{tikz}
\usetikzlibrary{calc,positioning,shapes,shadows,arrows,fit}

\theoremstyle{plain}
\newtheorem{theorem}{Theorem}
\newtheorem{lemma}{Lemma}

\theoremstyle{definition}
\newtheorem{proposition}{Proposition}
\newtheorem{definition}{Definition}
\newtheorem{remark}{Remark}
\newtheorem{assumption}{Assumption}
\newtheorem{example}{Example}
\newtheorem{problem}{Problem}
 
\usepackage{algorithm}
\usepackage{algpseudocode}
\usepackage{pifont}

\newcommand{\T}{\mathcal{T}} 
\newcommand{\A}{\mathcal{A}} 
\newcommand{\init}{\mathit{init}}

\newcommand{\AP}{\mathit{AP}} 

\renewcommand{\baselinestretch}{0.98}  

\usepackage{caption}

\title{\LARGE \bf
Cooperative Planning for Coupled Multi-Agent Systems under Timed Temporal Specifications 
}

\author{Alexandros Nikou, Dimitris Boskos, Jana Tumova and Dimos V. Dimarogonas
\thanks{The authors are with the ACCESS Linnaeus Center, School of Electrical
Engineering, KTH Royal Institute of Technology, SE-100 44, Stockholm,
Sweden and with the KTH Center for Autonomous Systems. Email: {\tt\small \{anikou, boskos, tumova, dimos\}@kth.se}. This work was supported by the H2020 ERC Starting Grant BUCOPHSYS, the Swedish Research Council (VR) and the Knut och Alice Wallenberg Foundation.}
}

\begin{document}

\maketitle
\thispagestyle{empty}
\pagestyle{empty}

\begin{abstract}
This paper presents a fully automated procedure for controller synthesis for multi-agent systems under coupled constraints. Each agent has dynamics consisting of two terms: the first one models the coupled constraints and the other one is an additional  control input. We aim to design these inputs so that each agent meets an individual high-level specification given as a Metric Interval Temporal Logic (MITL). First, a decentralized abstraction that provides a time and space discretization of the multi-agent system is designed. Second, by utilizing this abstraction and techniques from formal verification, we provide an algorithm that computes the individual runs which provably satisfy the high-level tasks. The overall approach is demonstrated in a simulation example.
\end{abstract}

\section{Introduction}

Cooperative control of multi-agent systems has traditionally focused on designing distributed control laws in order to achieve global tasks such as consensus, formation and rendez-vous (\cite{ren_beard_consensus, olfati_murray_concensus, jadbabaie_morse_coordination, shi2013robust, tanner_flocking}) and at the same time fulfill properties such as network connectivity (\cite{egerstedt_formation, zavlanos_2008_distributed}). Over the last few years,  multi-agent control under complex high-level specifications has been gaining significant attention. In particular, coordination of multi-robot teams under qualitative temporal tasks constitutes an emerging application in this field. In this work, we aim to additionally introduce specific time bounds into these tasks, in order to include specifications such as ``Visit region $A$ within 5 time units" or ``Periodically survey regions $A_1$, $A_2$, $A_3$, avoid region $X$ and always keep the longest time between two consecutive visits to $A_1$ below 20 time units".

The specification language that has primarily been used to express the tasks is Linear Temporal Logic (LTL) (see, e.g., \cite{loizou_2004}). LTL has been proven a valuable tool for controller synthesis, because it provides a compact mathematical formalism for specifying desired behaviors of a system. There is a rich body of literature containing algorithms for verification and synthesis of multi-agent systems under temporal logic specifications (\cite{guo_2015_reconfiguration, frazzoli_vehicle_routing}). A common approach in multi-agent planning under LTL specifications is the consideration of a centralized, global task for the team, which is then decomposed into local tasks to be accomplished by the individual agents (see \cite{belta_regular2, belta_cdc_reduced_communication}). A three-step hierarchical procedure to address this problem is described as follows (\cite{fainekos_planning}): first, the robot dynamics is abstracted into a discrete transition system using sampling or cell decomposition methods based on triangular, rectangular or other partitions. Second, invoking ideas from formal verification, a discrete plan that meets the high-level task is synthesized. Third, the discrete plan is translated into a sequence of continuous controllers for the original system.

Time constraints in the system modeling have been considered e.g., in \cite{belta_optimality, quottrup_timed_automata, belta_2011_timed_automata}. Both the aforementioned, as well as most existing works on multi-agent planning, consider temporal properties which treat time in a qualitative manner. However, for real applications, a multi-agent team might be required to perform a specific task within a certain time bound, rather than at some arbitrary time in the future, i.e. in a quantitative manner.  Timed specifications have been considered in \cite{liu_MTL, murray_2015_stl, baras_MTL_2016, topcu_2015, fainekos_mtl_2015_robot}. However, all these works are restricted to single agent planning and are not extendable to multi-agent systems in a straightforward way.

The multi-agent case has been considered in \cite{frazzoli_MTL}, where the vehicle routing problem was addressed, under Metric Temporal Logic (MTL) specifications. The corresponding approach does not rely on automata-based verification, as it is based on a construction of linear inequalities and the solution of a Mixed-Integer Linear Programming (MILP) problem. An automata-based solution was proposed in our previous work \cite{alex_2016_acc}, where Metric Interval Temporal Logic (MITL) formulas were introduced in order to synthesize controllers such that every agent fulfills an individual specification and the team of agents fulfill a global specification.

In \cite{alex_2016_acc}, the abstraction of the dynamics was given and an upper bound of the time that each agent needs to finish a transition from one region to another was assumed. Furthermore, potential coupled constraints between the agents were not taken into consideration. In this work, we aim to address the aforementioned issues. The dynamics of each agent consists of two parts: the first part is a consensus type term representing the coupling between the agent and its neighbors, and the second one is an additional control input which will be exploited for high-level planning. Hereafter, it will be called a free input. A decentralized abstraction procedure is provided, which leads to an individual Transition System (TS) for each agent and provides the basis for high-level planning. Additionally, this abstraction is associated to a time quantization which allows us to assign precise time durations to the transitions of each agent.

There is a rich literature on abstractions for dynamical systems (see e.g., \cite{alur_2000_discrete_abstractions, zamani_2012_symbolic, belta_2013_time-constraints, liu_2016_abstraction}). Multi-agent abstractions have been addressed in \cite{abate_2014_finite_abstractions, rungger_compositional_abstractions_2015, tabuada_2015_compositional, pola_2016_symbolic,  boskos_cdc_2015}. Motivated by \cite{boskos_cdc_2015}, we start from the dynamics of each agent and we construct a TS for each agent in a decentralized manner. An individual task is assigned to each agent and we aim to design the free inputs so that each agent performs the desired individual task within specific time bounds. To the best of the authors' knowledge, this is the first time that a fully automated framework for multi-agent systems consisting of both constructing an abstraction and conducting high-level timed temporal logic planning is considered. Hence, this works lies in the intersection of the fields of multi-agent systems, abstractions and timed formal verification.

The contribution of this paper is to provide an automatic controller synthesis method of a general framework of coupled multi-agent systems under high-level tasks with timed constraints. Compared to the existing works on multi-agent planning under temporal logic specifications, the proposed approach yields the first solution to the problem of planning of dynamically coupled multi-agent systems  under timed temporal specifications in a distributed way.

The remainder of the paper is structured as follows. In Sec. \ref{sec: preliminaries} a description of the necessary mathematical tools, the notations and the definitions are given. Sec. \ref{sec: prob_formulation} provides the dynamics of the system and the formal problem statement. Sec. \ref{sec: solution} discusses the technical details of the solution. Sec. \ref{sec: simulation_results} is devoted to a simulation example. Finally, the conclusions and the future work directions are discussed in Sec. \ref{sec: conclusions}.

\section{Notation and Preliminaries} \label{sec: preliminaries}

\subsection{Notation}

We denote by $\mathbb{R}, \mathbb{Q}_+, \mathbb{N}$ the set of real, nonnegative rational and natural numbers including 0, respectively. Also, define $\mathbb{T}_{\infty} = \mathbb{T} \cup \{\infty\}$ for a set $\mathbb{T} \subseteq \mathbb{R}$. Given a set $S$, we denote by $|S|$ its cardinality and by $2^S$ the set of all its subsets. For a subset $S$ of $\mathbb{R}^n$, we denote by $\text{cl}(S), \text{int}(S)$ and $\partial S = \text{cl}(S) \backslash \text{int}(S)$ its closure, interior and boundary, respectively and $\backslash$ is used for set subtraction. The notation $\|x\|$ is used for the Euclidean norm of a vector $x \in \mathbb{R}^n$ and $\|A\| = \text{max} \{\|Ax\| : \|x\| = 1\}$ for the induced norm of a matrix $A \in \mathbb{R}^{m \times n}$. Given a matrix $A$, the spectral radius of $A$ is denoted by $\lambda_{\text{max}}(A) = \text{max} \{|\lambda| : \lambda \in \sigma(A) \}$, where $\sigma(A)$ is the set of all the eigenvalues of $A$.

\subsection{Multi-Agent Systems}
\label{sec:prelims:system}

Consider a set of agents $\mathcal{I} = \{ 1,2, \ldots, N\}$ operating in $\mathbb{R}^n$. The topology of the multi-agent network is modeled through a static undirected graph $\mathcal{G} = (\mathcal{I},\mathcal{E})$, where $\mathcal{I}$ is the set of nodes (agents) and $\mathcal{E} \subseteq \{ \{i,j\} : i,j \in \mathcal I, i \neq j\}$ is the set of edges (denoting the communication capability between neighboring respective agents).  For each agent, its neighbors' set $\mathcal{N}(i)$ is defined as $\mathcal{N}(i) = \{j_1, \ldots, j_{N_i} \} = \{ j \in \mathcal{I} : \{i,j\} \in \mathcal{E}\}$ where $N_i = |\mathcal N(i)|$.

Given a vector $x_i = (x^1_i, \ldots, x^n_i) \in \mathbb{R}^n$, the component operator $c(x_i, k) = x_i^k \in \mathbb{R}, k = 1, \ldots, n$ gives the projection of $x_i$ onto its $k$-th component (see \cite{mesbahi2010graph}). Similarly, for the stack vector $x = (x_1, \ldots, x_N) \in \mathbb{R}^{Nn}$ the component operator is defined as $c(x, k) = (c(x_1, k), \ldots, c(x_N, k))  \\ \in \mathbb{R}^{N}, k = 1, \ldots, n$. By using the component operator, the norm of a vector $x \in \mathbb{R}^{Nn}$ can be computed as $\|x\| = \left\{ \sum_{k = 1}^{n} \| c(x, k) \|^2 \right\}^{\frac{1}{2}}$.

The Laplacian matrix $L(\mathcal{G}) \in \mathbb{R}^{N \times N}$ of the graph $\cal{G}$ is defined as
$ L(\mathcal{G}) = D(\mathcal{G}) D(\mathcal{G})^{\tau}$ where $D(\mathcal{G})$ is the $N \times |\mathcal{E}|$ incidence matrix (\cite{mesbahi2010graph}). The graph Laplacian $L(\mathcal{G})$ is positive semidefinite and symmetric. By considering an ordering $0 = \lambda_1(\mathcal{G}) \leq \lambda_2(\mathcal{G}) \leq \ldots \leq \lambda_N(\mathcal{G}) = \lambda_{\text{max}}(\mathcal{G})$ of the eigenvalues of $L(\mathcal{G})$ then we have that $\lambda_2(\mathcal{G}) > 0$ iff $\mathcal{G}$ is connected (\cite{mesbahi2010graph}).

We denote by $\tilde{x} \in \mathbb{R}^{|\mathcal{E}|n}$ the stack column vector of the vectors $x_i-x_j, \{i, j\} \in \mathcal{E}$ with the edges ordered as in the case of the incidence matrix. Thus, the following holds:
\begin{equation} \label{eq:xtilde}
\tilde{x} = D(\mathcal{G})^{\tau} x.
\end{equation}

\subsection{Cell Decompositions}

In the subsequent analysis a discrete partition of the workspace into cells will be considered which is formalized through the following definition. 

\begin{definition} \label{def: cell_decomposition}
A \textit{cell decomposition} $S = \{S_\ell\}_{\ell \in \mathbb{I}}$ of a set $\mathcal{D} \subseteq \mathbb{R}^n$, where $\mathbb{I} \subseteq \mathbb{N}$ is a finite or countable index set, is a family of uniformly bounded convex sets $S_\ell, \ell \in \mathbb{I}$ such that $\text{int}(S_\ell) \cap \text{int}(S_{\hat{\ell}}) = \emptyset$ for all $\ell, \hat{\ell} \in \mathbb{I}$ with $\ell \neq \hat{\ell}$ and $\cup_{\ell \in \mathbb{I}} S_\ell = \mathcal D$.
\end{definition}

\begin{example} \label{ex: example_0}
An example of a cell decomposition with $\mathbb{I} = \{1, 2, 3, 4, 5,6\}$ and $S = \{S_\ell\}_{\ell \in \mathbb{I}} 
= \{S_1, S_2, S_3, S_4, S_5, S_6\}$ is depicted in Fig. \ref{fig: example_0}. This cell decomposition will be used as reference for the following examples.
\begin{figure}[ht!]
	\centering
	\begin{tikzpicture}[scale = 0.5]
	
	\draw[step=2.5, line width=.04cm] (-2.5, -5.0) grid (0,0);
	\draw[line width=.04cm] (-7.5,0.0) -- (-2.5,0.0);
	\draw[line width=.04cm] (-7.5,-2.5) -- (-2.5,-2.5);
	\draw[line width=.04cm] (-7.5,-5.0) -- (-2.5,-5.0);
	\draw[step=2.5, line width=.04cm] (-10.0, -5.0) grid (-7.5,0);
	
	\node at (-8.7, -1.2) {$S_1$};
	\node at (-5.0, -1.2) {$S_2$};
	\node at (-1.3, -1.2) {$S_3$};
	\node at (-1.3, -3.7) {$S_4$};
	\node at (-5.0, -3.7) {$S_5$};
	\node at (-8.7, -3.7) {$S_6$};
	\end{tikzpicture}
	
	\caption{An example of a cell decomposition with $|\mathbb{I}| = 6$ cells}
	\label{fig: example_0}
\end{figure}
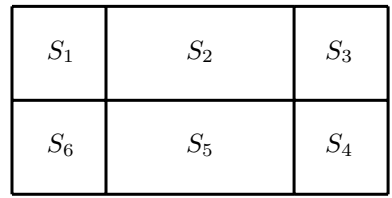
\end{example}

\subsection{Time Sequence, Timed Run and Weighted Transition System}

In this section we review some basic definitions from computer science that are required in the sequel. 

An infinite sequence of elements of a set $X$ is called an \textit{infinite word} over this set and it is denoted by $\chi = \chi(0)\chi(1) \ldots$ The $i$-th element of a sequence is denoted by $\chi(i)$.

\begin{definition} (\cite{alur1994}) A \textit{time sequence} $\tau = \tau(0) \tau(1) \ldots$ is an infinite sequence of time values $\tau(j) \in \mathbb{T} = \mathbb{Q}_{+}$, satisfying the following properties:
	\begin{itemize}
		\item Monotonicity: 
		$\tau(j) < \tau(j+1)$ for all $j \geq 0$.
		\item Progress: For every $t \in \mathbb{T}$, there exists 	$\ j \ge 1$, such that $\tau(j) > t$.
	\end{itemize}
\end{definition}

An \textit{atomic proposition} $p$ is a statement 
that is either True $(\top)$ or False $(\bot)$. 

\begin{definition} (\cite{alur1994})
	Let $\AP$ be a finite set of atomic propositions. A \textit{timed word} $w$ over the set $\AP$ is an infinite sequence $w^t = (w(0), \tau(0)) (w(1), \tau(1)) \ldots$ where $w(0) w(1) \ldots$ is an infinite word over the set $2^{\AP}$ and $\tau(0) \tau(1) \ldots$ is a time sequence with $\tau(j) \in \mathbb{T}, \ j \geq 0$.
\end{definition}

\begin{definition} \label{def: WTS}
	A Weighted Transition System (\textit{WTS}) is a tuple $(S, S_0, Act, \longrightarrow, d, AP, L)$ where
	$S$ is a finite set of states;
	$S_0 \subseteq S$ is a set of initial states;
	$Act$ is a set of actions;
	$\longrightarrow \subseteq S \times Act \times S$ is a transition relation;
	$d: \longrightarrow \rightarrow \mathbb{T}$ is a map that assigns a positive weight (time values in this framework) to each transition;
	$\AP$ is a finite set of atomic propositions; and
	$L: S \rightarrow 2^{AP}$ is a labeling function.
	For simplicity, the notation $s \overset{\alpha}{\longrightarrow} s'$ is used to denote that $(s, \alpha, s') \in \longrightarrow$ for $s, s' \in S$ and $\alpha \in Act$. Furthermore, for every $s \in S$ and $\alpha \in Act$ the operator $\text{Post}(s, \alpha) = \{s' \in S : (s, \alpha, s') \in \longrightarrow\}$ is defined.
\end{definition}

\begin{definition}\label{run_of_WTS}
	A \textit{timed run} of a WTS is an infinite sequence $r^t = (r(0), \tau(0))(r(1), \tau(1)) \ldots$,
	such that $r(0) \in S_0$, and for all $j \geq 1$, it holds that $r(j) \in S$ and $(r(j), \alpha(j), r(j+1)) \in \longrightarrow$ for a sequence of actions $\alpha(1) \alpha(2) \ldots$ with $\alpha(j) \in Act, \forall \ j \geq 1$. The \textit{time stamps} $\tau(j), j \geq 0$ are inductively defined as
	\begin{enumerate}
		\item $\tau(0) = 0$.
		\item $\displaystyle \tau(j+1) =  \tau(j) + d(r(j), r(j+1)), \ \forall \ j \geq 1.$
	\end{enumerate}
	Every timed run $r^t$ generates a \textit{timed word}
	$w(r^t) = 
	(w(0), \tau(0)) \ (w(1), \tau(1)) \ldots$
	over the set $2^{\AP}$ where $w(j) = L(r(j))$, $\forall \ j \geq 0$ is the subset of atomic propositions that are true at state $r(j)$. 
	
\end{definition}

\subsection{Metric Interval Temporal Logic}

The syntax of \textit{Metric Interval Temporal Logic (MITL)} over a set of atomic propositions $AP$ is defined by the grammar
\begin{equation} \label{eq: grammar}
	\varphi := p \ | \ \neg \varphi \ | \ \varphi_1 \wedge \varphi_2 \ | \ \bigcirc_I \varphi  \ | \ \Diamond_I \varphi \mid \square_I \varphi \mid  \varphi_1 \ \mathcal{U}_I \ \varphi_2
\end{equation}
where $p \in \AP$, and $\bigcirc$, $\Diamond$, $\square$ and $\mathcal U$ are the next, eventually, always and until temporal operator, respectively. $I \subseteq \mathbb{T}$ is a non-empty time interval in one of the following forms: $[i_1, i_2], [i_1, i_2),(i_1, i_2], $ $ (i_1, i_2), [i_1, \infty], (i_1, \infty)$ where $i_1, i_2 \in \mathbb{T}$ with $i_1 < i_2$. MITL can be interpreted either in continuous or point-wise semantics \cite{pavithra_expressiveness}. The latter approach is utilized, since the consideration of point-wise (event-base) type semantics renders a framework that includes Transition Systems and automata construction, namely our current approach, more natural. The MITL formulas are interpreted over timed runs such as the ones produced by a WTS (Def.~\ref{run_of_WTS}).

\begin{definition} \label{def:mitl_semantics} (\cite{pavithra_expressiveness}, \cite{quaknine_decidability})
	Given a timed word $w^t = (w(0),\tau(0))(w(1),\tau(1)) \dots$ and an MITL formula $\varphi$, we define $(w^t, i) \models \varphi$, for $\ i \geq 0$ (read $w^t$ satisfies $\varphi$ at position $i$) as follows:
	\begin{align*} \label{eq: for1}
		(w^t, i) &\models p \Leftrightarrow p \in w(i) \\
		(w^t, i) &\models \neg \varphi \Leftrightarrow (w^t, i) \not \models \varphi \\
		(w^t, i) &\models \varphi_1 \wedge \varphi_2 \Leftrightarrow (w^t, i) \models \varphi_1 \ \text{and} \ (w^t, i) \models \varphi_2 \\
		(w^t, i) &\models \bigcirc_I \ \varphi \Leftrightarrow (w^t, i+1) \models \varphi \ \text{and} \    \tau(i+1) - \tau(i) \in I\\
		(w^t, i) & \models \Diamond_I \varphi \Leftrightarrow \exists j \ge i, \ \text{s.t. } (w^t, j) \models \varphi, \tau(j)-\tau(i) \in {I} \\
		(w^t, i) & \models \square_I \varphi \Leftrightarrow \forall j \ge i, \ \tau(j)-\tau(i) \in {I} \Rightarrow (w^t, j) \models \varphi  \\
		(w^t, i) &\models \varphi_1 \ \mathcal{U}_I \ \varphi_2 \Leftrightarrow \exists j\ge i, \ \text{s.t. } (w^t, j) \models \varphi_2, \\ &\tau(j)-\tau(i) \in I \ \text{and } (w^t, k) \models \varphi_1 \ \text{for every} \ i \leq k < j.
	\end{align*}
	
\end{definition}

It has been proved that MITL is decidable in both finite and infinite words \cite{alur_mitl} and in both pointwise and continuous semantics \cite{reynold}. The model checking and satisfiability problems are \textit{EXPSPACE}-complete.

\begin{example}
\begin{figure}[ht!]
   	\centering
   	\begin{tikzpicture}[scale = 1.4]       
   	\node(pseudo1) at (-0.8,0){};
   	\node(0) [line width = 1.0] at (0,0)[shape=circle,draw][fill=green!20]          {$s_0$};
   	\node(1) [line width = 1.0] at (2.0,0)[shape=circle,draw][fill=blue!20]         {$s_1$};
   	\node(5) [line width = 1.0] at (4.0,0)[shape=circle,draw][fill=blue!20]         {$s_2$};
   	
   	\path [->] [line width = 1.0]
   	(0)     edge       [bend left = 15]              node  [above]  {$1.0$}  (1)
   	(1)     edge     [bend right = -15]           node  [below]  {$2.0$}     (0)
   	(1)     edge     [bend right = -15]           node  [above]  {$1.5$}     (5)
   	(5)     edge     [bend right = -15]           node  [below]  {$0.5$}     (1)

   	(pseudo1) edge                                       (0);
   	
   	\end{tikzpicture}
   	\caption{An example of a WTS}
   	\label{fig:wtsexample}
   \end{figure}
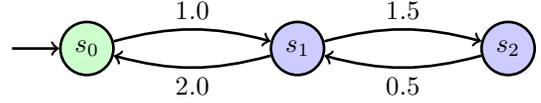
Consider the WTS $\mathcal{T}$ with $S = \{s_0, s_1, s_2\}, S_0 = \{s_0\}, Act = \emptyset,  \longrightarrow = \{(s_0, \emptyset, s_1), \\ (s_1, \emptyset, s_2), (s_1, \emptyset, s_0), (s_2, \emptyset, s_1)\}, d((s_0, \emptyset, s_1)) = 1.0, d((s_1, \emptyset, s_2)) = 1.5, d((s_1, \emptyset, s_0)) = 2.0, \\ d((s_2, \emptyset, s_1)) = 0.5, AP = \{green\}, L(s_0) = \{green\}, L(s_1) = L(s_2) = \emptyset$ depicted in Fig. \ref{fig:wtsexample}.
   
Let two timed runs of the system: $r_1^t = (s_0, 0.0)(s_1, 1.0)(s_0, 3.0)(s_1, 4.0) \ldots, r_2^t = (s_0, 0.0)(s_1, 1.0)(s_2, 2.5)(s_1, 3.0) \ldots$ and two MITL formulas $\varphi_1 = \Diamond_{[2, 5]} \{green\}, \varphi_2 = \square_{[0, 5]} \{green\}$. According to the MITL semantics, it can be seen that the timed run $r_1^t$ satisfies the formula $\varphi_1$ (we formally write $r_1^t \models \varphi_1$), since at the time stamp $3.0 \in [2, 5]$ we have that $L(s_0) = \{green\}$ so the atomic proposition $green$ occurs at least once in the given interval. On the other hand, the timed run $r_2^t$ does not satisfy the formula $\varphi_2$ (we formally write $r_2^t \not \models \varphi_2$) since the atomic proposition $green$ does not always hold at every time stamp of the runs (it holds only at the time stamp $0.0$).
\end{example}

\subsection{Timed B\"uchi Automata} \label{sec: timed_automata}
\textit{Timed B\"uchi Automata (TBA)} were introduced in \cite{alur1994}. In this work, the notation from \cite{bouyer_phd, tripakis_tba} is partially adopted. 
Let $C = \{c_1, \ldots, c_{|C|}\}$ be a finite set of \textit{clocks}. The set of \textit{clock constraints} $\Phi(C)$ is defined by the grammar
\begin{equation}
	\phi :=  \top \mid \ \neg \phi \ | \ \phi_1 \wedge \phi_2 \ | \ c \bowtie \psi \
\end{equation}
where $c \in C$ is a clock, $\psi \in \mathbb{T}$ is a clock constant and $\bowtie \ \in  \{ <, >, \geq, \leq, = \}$. An example of clock constraints for a set of clocks $C = \{c_1, c_2\}$ can be $\Phi(C) = \{ c < c_1 \vee c > c_2\}$. A clock \textit{valuation} is a function $\nu: C \rightarrow\mathbb{T}$ that assigns a value to each clock. A clock $c_i$ has valuation $\nu_i$ for $i \in \{1, \ldots, |C|\}$, and $\nu = (\nu_1, \ldots, \nu_{|C|})$. By $\nu \models \phi$ is denoted the fact that the valuation $\nu$ satisfies the clock constraint $\phi$.


\begin{definition}
	A \textit{Timed B\"uchi Automaton} is a tuple $\mathcal{A} = (Q, Q^{\text{init}}, C, Inv,
	E, F, AP, \mathcal{L})$ where
	$Q$ is a finite set of locations;
	$Q^{\text{init}} \subseteq Q$ is the set of initial locations;
	$C$ is a finite set of clocks;
	$Inv: Q \rightarrow \Phi(C)$ is the invariant; 
	$E \subseteq Q \times \Phi(C) \times 2^C \times Q$ gives the set of edges;
	$F \subseteq Q$ is a set of accepting locations;
	$\AP$ is a finite set of atomic propositions; and
	$\mathcal{L}: Q \rightarrow 2^{AP}$ labels every state with a subset of atomic propositions.
\end{definition}

A state of $\mathcal{A}$ is a pair $(q, \nu)$ where $q \in Q$ and $\nu$ satisfies the \textit{invariant} $Inv(q)$, i.e., $\nu \models Inv(q)$. The initial state of $\mathcal{A}$ is $(q(0), (0,\ldots,0))$, where $q(0) \in Q^{\text{init}}$. Given two states $(q, \nu)$ and $(q', \nu')$ and an edge $e = (q, \gamma, R, q')$, there exists a \textit{discrete transition} $(q, \nu) \overset{e}{\longrightarrow} (q', \nu')$ iff $\nu \models \gamma$, $\nu' \models Inv(q')$, and $R$ is the \textit{reset set}, i.e., $\nu'_i = 0$ for $c_i \in R$ and $\nu'_i = \nu_i$ for $c_i \notin R$. Given a $\delta \in \mathbb{T}$, there exists a \textit{time transition} $(q, \nu) \overset{\delta}{\longrightarrow} (q', \nu')$ iff $q = q', \nu' = \nu+\delta$ ($\delta$ is summed component-wise) and $\nu' \models Inv(q)$. We write $(q, \nu) \overset{\delta}{\longrightarrow} \overset{e}{\longrightarrow} (q', \nu')$ if there exists $q'', \nu''$ such that $(q, \nu) \overset{\delta}{\longrightarrow} (q'', \nu'')$ and $(q'', \nu'') \overset{e}{\longrightarrow} (q', \nu')$ with $q'' = q$.

An infinite run of $\mathcal{A}$ starting at state $(q(0), \nu)$ is an infinite sequence of time and discrete transitions $(q(0), \nu(0))\overset{\delta_0}{\longrightarrow} (q(0)', \nu(0)') \overset{e_0}{\longrightarrow} (q(1), \nu(1)) \overset{\delta_1}{\longrightarrow} (q(1)', \nu(1)') \ldots$, where $(q(0),\nu(0))$ is an initial state. 
This run produces the timed word $w = (\mathcal{L}(q(0)), \tau(0)) (\mathcal{L}(q(1)), \tau(1)) \ldots$ with $\tau(0) = 0$ and $\tau(i+1) = \tau(i) +\delta_i$,  $\forall \ i \geq 1$. The run is called \textit{accepting} if $q(i) \in F$ for infinitely many times. A timed word is \textit{accepted} if there exists an accepting run that produces it. The problem of deciding the emptiness of the language of a given TBA $\mathcal{A}$ is PSPACE-complete \cite{alur1994}. In other words, an accepting run of a given TBA $\mathcal{A}$ can be synthesized, if one exists.

Any MITL formula $\varphi$ over $AP$ can be algorithmically translated to a TBA with the alphabet $2^\AP$, such that the language of timed words that satisfy $\varphi$ is the language of timed words produced by the TBA (\cite{alur_mitl, maler_MITL_TA, nickovic_timed_aut}).

\begin{example} \label{ex: example_1}
	The TBA $\mathcal{A}$ with $Q = \{q_0, q_1, q_2\}, Q^{\text{init}} = \{q_0\}, C = \{c\}, Inv(q_0) = Inv(q_1) = Inv(q_2) = \emptyset, E = \{(q_0, \{c \leq c_2\}, \emptyset, q_0), (q_0, \{c \leq c_1 \vee c > c_2\}, c, q_2), (q_0, \{c \geq c_1 \wedge c \leq c_2\}, c, q_1), (q_1, \top, c, q_1), (q_2, \top, c, q_2)\}, F = \{q_1\}, AP = \{green\}, \mathcal{L}(q_0) = \mathcal{L}(q_2) = \emptyset, \mathcal{L}(q_1) = \{green\} $ that accepts all the timed words that satisfy the formula $\varphi = \Diamond_{[c_1, c_2]} \{green\}$ is depicted in Fig. \ref{fig:TBA_example}. This formula will be used as reference for the following examples and simulations.

	\begin{figure}[ht!]
		\centering
		\begin{tikzpicture}[scale = 1.0]
		\node(pseudo) at (-1.5,0){};
		\node(0) [line width = 1.0] at (0,0)[shape=circle,draw][fill=orange!20] {$\ q_0 \ $};
		\node(1) [line width = 1.0] at (4.5,0)[shape=circle,draw, double][fill=green!20] {$\ q_1 \ $};
		\node(2) [line width = 1.0] at (0,-2)[shape=circle,draw][fill=orange!20] {$\ q_2 \ $};

		\path [->] [line width = 1.0]
		(0)      edge [bend left = 0]         node [above]  {}     (1)
		(0)      edge [loop above]    node [above]  {}     ()
		(1)      edge [loop above]    node [above]  {$\top, c := 0$}     ()
		(2)      edge [loop below]    node [below]  {$\top, c := 0$}     ()
		(0)      edge [bend left = 0]         node [above]  {}     (2)

		(pseudo) edge                                       (0);
		
		\node at (-0.2, 1.6) {$c \leq c_2, \emptyset$};
		\node at (2.2, 0.3) {$c \geq c_1 \wedge c \leq c_2$};
		\node at (2.3, -0.3) {$\tiny c := 0$};
		\node at (-2.0, -0.7) {$\tiny c < c_1 \vee c > c_2$};
		\node at (-1.9, -1.1) {$\tiny c := 0$};
		
		\node at (4.5,-0.8) {$\{green\}$};		
		\end{tikzpicture}
		\caption{A TBA $\mathcal{A}$ that accepts the runs that satisfy formula $\varphi = \Diamond_{[c_1, c_2]} \{green\}$.}
		\label{fig:TBA_example}
	\end{figure}
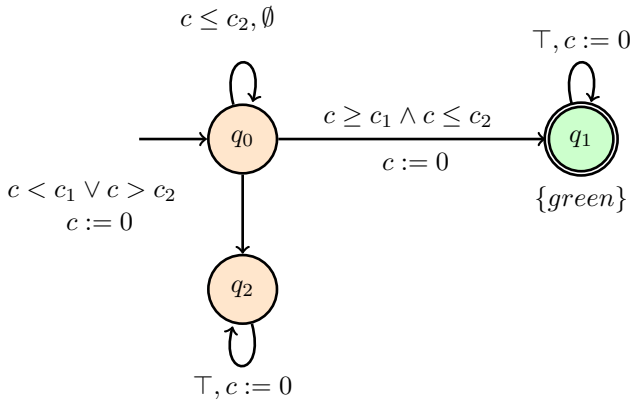
	
	An example of a timed run of this TBA is $(q_0, 0) \overset{\delta = \alpha_1}{\longrightarrow} (q_0, \alpha_1) \overset{e = (q_0, \{c \geq c_1 \wedge c \leq c_2\}, c, q_1)}{\longrightarrow} (q_1, 0) \ldots$ with $c_1 \leq \alpha_1 \leq c_2 $, which generates the timed word $w^t = (\mathcal{L}(q_0), 0)(\mathcal{L}(q_0), \alpha_1)(\mathcal{L}(q_1), \alpha_1) \ldots = (\emptyset, 0)(\emptyset, \alpha_1)(\{green\}, \alpha_1) \ldots$ that satisfies the formula $\varphi$. The timed run $(q_0, 0) \overset{\delta = \alpha_2}{\longrightarrow} (q_0, \alpha_2) \\ \overset{e = (q_0, \{c \leq c_1 \vee c > c_2\}, c, q_2)}{\longrightarrow} (q_2, 0) \ldots$ with $\alpha_2 < c_1$, generates the timed word $w^t = (\mathcal{L}(q_0), 0)(\mathcal{L}(q_0), \alpha_2) \\ (\mathcal{L}(q_2), \alpha_2) \ldots = (\emptyset, 0)(\emptyset, \alpha_2)(\emptyset, \alpha_2) \ldots$ that does not satisfy the formula $\varphi$.
\end{example}

\section{Problem Formulation} \label{sec: prob_formulation}

\subsection{System Model}
We focus on multi-agent systems with coupled dynamics of the form
\begin{equation} \label{eq: system}
	\dot{x}_i = -\sum_{j \in \mathcal{N}(i)}^{} (x_i - x_j)+v_{i}, \ x_i \in \mathbb{R}^n, \  i \in \mathcal{I}.
\end{equation}
where $v_i \in \mathbb{R}^n, \ i \in \mathcal{I}$. The dynamics \eqref{eq: system} consists of two parts; the first part is a consensus protocol representing the coupling between the agent and its neighbors, and the second one is a control input which will be exploited for high-level planning and is called free input. In this work, it is assumed that the free inputs are bounded by a positive constant $v_{\text{max}}$. Namely, $\| v_i(t) \| \leq v_{\text{max}}, \ \forall \ i \in \mathcal{I}, t \geq 0$.

\begin{assumption}
We assume that the communication graph $\mathcal{G} = (\mathcal{I},\mathcal{E})$ of the system is undirected and static i.e., every agent preserves the same neighbors for all times.
\end{assumption}

Notice that the system \eqref{eq: system} can be also expressed in the form $c(\dot{x}, k) = -L(\mathcal{G}) \ c(x, k)+c(v, k), k \in  \{1,\ldots,n\}$ where $x, v \in \mathbb{R}^{Nn}$ are obtained by invoking the definition of the component operator from Sec. \ref{sec:prelims:system}.

\subsection{Specification}
Our goal is to control the multi-agent system \eqref{eq: system} so that each agent obeys a given individual specification. In particular, it is required to drive each agent to a sequence of desired subsets of the \textit{workspace} $\mathbb R^n$ within certain time limits and provide certain atomic tasks there. Atomic tasks are captured through a finite set of \textit{services} $\Sigma_i, i \in \mathcal{I}$. Hence, it is desired to relate the position $x_i$ of each agent $i \in \mathcal{I}$ in the workspace with the services that are offered at $x_i$. Initially, a labeling function 
\begin{equation} \label{eq:label_lambda}
\Lambda_i:\mathbb{R}^n\to 2^{\Sigma_i}
\end{equation}
is introduced for each agent $i \in \mathcal{I}$ which maps each state $x_i \in \mathbb{R}^n$ to the subset of services $\Lambda_i(x_i)$ which hold true at $x_i$ i.e., the subset of services that agent $i$ \textit{can} provide in position $x_i$. It should be noted that although the term labeling function it is used, these functions should not be confused with the labeling functions of a WTS as in Definition \ref{def: WTS}. The union of all the labeling functions as $\Lambda(x) = \bigcup_{i \in \mathcal I} \Lambda_i(x)$ is also defined. 

Without loss of generality, we assume that $\Sigma_i \cap \Sigma_j = \emptyset$, for all $i,j \in \mathcal{I}, i \neq j$ which means that the agents do not share any services. Let us now introduce the following assumption which is important for defining the problem properly.
\begin{assumption}  \label{assumption: AP_cell_decomposition}
There exists a partition $S = \{S_\ell\}_{\ell \in \mathbb I}$ of the workspace which forms a cell decomposition according to Definition \ref{def: cell_decomposition} and respects the labeling function $\Lambda$ i.e., for all $S_\ell \in S$ it holds that $\Lambda(x) = \Lambda(x'), \forall \ x, x' \in S_\ell$.
This assumption intuitively and again without loss of generality, means that the same services hold at all the points that belong to the same cell of the partition.
\end{assumption}
Define now for each agent $i \in \mathcal{I}$ a labeling function 
\begin{equation} \label{eq:label_lambda2}
\mathcal{L}_i: S \to 2^{\Sigma_i}
\end{equation}
which denotes the fact that when agent $i$ visits a region $S_\ell \in S$, it chooses to \textit{provide} a subset of services that are being offered there i.e., it chooses to satisfy a subset of $\mathcal{L}_i(S_\ell)$.

The trajectory of each agent $i$ is denoted by $x_i(t), t \geq 0, i \in \mathcal{I}$. The trajectory $x_i(t), i \in \mathcal{I}$ is associated with a unique sequence $r_{x_i}^t = (r_i(0), \tau_i(0))(r_i(1), \tau_1(1))(r_i(2), \tau_i(2))\ldots$ of regions that the agent $i$ crosses, where for all $j \ge 0$, $r_i(j) \in S_\ell$ for some $\ell \in \mathbb{I}$, $\Lambda_i(x_i(t)) = \mathcal{L}_i(r_i(j)), \forall \ t \in [\tau_i(j), \tau_i(j+1))$ and $r_i(j) \ne r_i(j+1)$. The equality $\Lambda_i(\cdot) = \mathcal{L}_i(\cdot), i \in \mathcal{I}$ is feasible due to assumption 2. The timed word $w_{x_i}^t = (w_i(0), \tau_i(0))(w_i(1), \tau_i(1))(w_i(2), \tau_i(2))\ldots$, where $w_i(j) = \mathcal{L}_i(r_i(j)), j \geq 0, i \in \mathcal{I}$, is associated uniquely with the trajectory $x_i(t)$ and represents the sequence of services that \textit{can be provided} by the agent $i$ following the trajectory $x_i(t), t \ge 0$. 

We define the \textit{timed service word} as
\begin{equation} \label{eq:time_service_word2}
\tilde{w}_{x_i}^t = (\beta_i(z_0), \tilde{\tau}_i(z_0))(\beta_i(z_1), \tilde{\tau}_i(z_1))(\beta_i(z_2), \tilde{\tau}_i(z_2))\ldots 
\end{equation} 
where $z_0 = 0 < z_1 < z_2 < \ldots$ is a sequence of integers, and for all $j \ge 0$ it holds that $\beta_i(z_j) \subseteq \mathcal{L}_i(r_i(z_j))$ and $\tilde{\tau}(z_j) \in [\tau_i(z_j), \tau_i(z_j+1))$. The timed service word is a sequence of services that are \textit{actually provided} by agent $i$ and it is compliant with the trajectory $x_i(t), t \ge 0$.

The specification task $\varphi_i$ given in MITL formulas over the set of services $\Sigma_i$ as in Definition \ref{def:mitl_semantics}, captures requirements on the services to be provided by agent $i$ for each $i \in \mathcal{I}$. We say that a trajectory $x_i(t)$ satisfies a given formula $\varphi_i$ in MITL over the set of atomic propositions $\Sigma_i$ if and only if there exits a \textit{timed service word}, as defined in \eqref{eq:time_service_word2}, that complies with $x_i(t)$ and satisfies $\varphi_i$ according to Definition \ref{def:mitl_semantics}.

\begin{example} \label{ex: example_01}
	We consider here an example in order to understand the notation and the technical terms that have been introduced until here. Let $N = 2$ agents performing in the partitioned environment of Fig. \ref{fig: example_01}. The agents have the ability to pick up, deliver and throw a different ball each. Let the services of each agent be $\Sigma_1 = \{\rm pickUp1, deliver1, throw1\}$ and $\Sigma_2 = \{\rm pickUp2, deliver2, throw2\}$. Note that $\Sigma_1 \cap \Sigma_2 \neq \emptyset$. We capture $3$ points of the trajectories of the agents that belong to different cells where different atomic propositions hold. Let $t_0 = t_0' = 0$ and $t_1 < t_1' < t_2 < t_2 < t_2' < t_3 < t_3'.$ The trajectories $x_1(t), x_2(t), t \ge 0$ are depicted with red lines. According to Assumption \ref{assumption: AP_cell_decomposition} we have that $S = \{S_\ell\}_{\ell \in \mathbb I} = \{S_1, \ldots, S_6\}$ where $\mathbb I = \{1, \ldots, 6\}$. We also have 
	\begin{align}
	&\Lambda_1(x_1(t)) = L_1(r_1(0)) = \{\rm pickUp1\}, t \in [0, t_1), \notag \\
	&\Lambda_1(x_1(t)) = L_1(r_1(1)) = \{\rm throw1\}, t \in [t_1, t_2), \notag \\ 
	&\Lambda_1(x_1(t)) = L_1(r_1(2)) = \{\rm deliver1\}, t \in [t_2, t_3), \notag \\
	&\Lambda_1(x_1(t)) = L_1(r_1(3)) = \emptyset, t \ge t_3. \notag \\
	&\Lambda_2(x_2(t)) = L_2(r_2(0)) = \{\rm pickUp2\}, t \in [0, t_1'), \notag \\
	&\Lambda_2(x_2(t)) = L_2(r_2(1))  = \{\rm deliver2\}, t \in [t_1', t_2'), \notag \\ 
	&\Lambda_2(x_2(t)) = L_2(r_2(2)) = \{\rm throw2\}, t \in [t_2', t_3'), \notag \\
	&\Lambda_2(x_2(t)) = L_2(r_2(3))  = \emptyset, t \ge t_3'. \notag
	\end{align}
	where generally $t_0 = t_0', t_1 \neq t_1'$ and $t_2 \neq t_2'$, $t_3 \neq t_3'$. By computing $w_i(j) = \mathcal{L}(r_i(j)), \forall \ i \in \{1,2\}, j \in \{1,2,3\}$ the corresponding individual timed words are given as:
	\begin{align}
	w^t_{x_1} &= (\{\rm pickUp1\}, t_0)(\{\rm throw1, t_1\})(\{\rm deliver1\}, t_2)(\emptyset, t_3) \notag \\
	w^t_{x_2} &= (\{\rm pickUp2\}, t_0')(\{\rm deliver2\}, t_1')(\{\rm throw2\}, t_2')(\emptyset, t_3'). \notag
	\end{align}
	According to \eqref{eq:time_service_word2}, two time service words (depicted with in Fig. \ref{fig: example_01}) are given as:
	\begin{align*}
	\tilde{w}_1^t &= (\beta_1(z_0), \tilde{\tau}_1(z_0))(\beta_1(z_1), \tilde{\tau}_1(z_1)) \notag \\ 
	\tilde{w}_2^t &= (\beta_2(z'_0), \tilde{\tau}_2(z'_0))(\beta_2(z'_1), \tilde{\tau}_2(z'_1)) \notag \\ 
	\end{align*}
	where for agent 1 we have: $z_0 = 0, z_1 = 2, \beta_1(z_0) = \{\rm pickUp1\} \subseteq \mathcal{L}_1(r_1(z_0)), \beta_1(z_1) = \{\rm deliver1\} \subseteq \mathcal{L}_1(r_1(z_1))$ and $\tilde{\tau}_1(z_0) = \tilde{t}_0' \in [\tau_1(z_0), \tau_1(z_0+1)) = [t_0, t_1), \tilde{\tau}_1(z_1) = \tilde{t}_1' \in [\tau_1(z_1), \tau_1(z_1+1)) = [t_2, t_3)$. The corresponding elements for agent 2 are $z'_0 = 0, z'_1 = 2, \beta_2(z'_0) = \{\rm pickUp2\} \subseteq \mathcal{L}_2(r_2(z'_0)), \beta_2(z'_1) = \{\rm deliver2\} \subseteq \mathcal{L}_2(r_2(z'_1))$ and $\tilde{\tau}_2(z'_0) = \tilde{t}_0'' \in [\tau_2(z'_0), \tau_2(z'_1)) = [t_0', t_1'), \tilde{\tau}_2(z_1') = \tilde{t}_1'' \in [\tau_1(z_1'), \tau_1(z_1'+1)) = [t_2', t_3')$.
	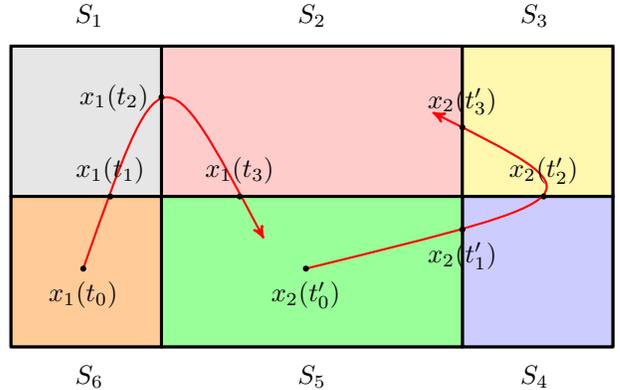
\begin{figure}[ht!]
		\centering
		\begin{tikzpicture}[scale = 0.8]
		
		\draw[step=2.5, line width=.04cm] (-2.5, -5.0) grid (0,0);
		\draw[line width=.04cm] (-7.5,0.0) -- (-2.5,0.0);
		\draw[line width=.04cm] (-7.5,-2.5) -- (-2.5,-2.5);
		\draw[line width=.04cm] (-7.5,-5.0) -- (-2.5,-5.0);
		\draw[step=2.5, line width=.04cm] (-10.0, -5.0) grid (-7.5,0);
		
		\filldraw[fill=black!10, line width=.04cm] (-10, -2.5) rectangle +(2.5, 2.5);
		\filldraw[fill=orange!40, line width=.04cm] (-10, -5.0) rectangle +(2.5, 2.5);
		\filldraw[fill=red!20, line width=.04cm] (-7.5, -2.5) rectangle +(5.0, 2.5);
		\filldraw[fill=yellow!40, line width=.04cm] (-2.5, -2.5) rectangle +(2.5, 2.5);
		\filldraw[fill=blue!20, line width=.04cm] (-2.5, -5.0) rectangle +(2.5, 2.5);
		\filldraw[fill=green!40, line width=.04cm] (-7.5, -5.0) rectangle +(5.0, 2.5);
		
		\draw [color=red,thick,->,>=stealth'](-8.8, -3.7) .. controls (-7.50, -0.0) .. (-5.8, -3.2);
		
		\draw [color=red,thick,->,>=stealth'](-5.1, -3.7) .. controls (-0.2, -2.5) .. (-3.0, -1.1);
		
		\draw (-8.8, -3.7) node[circle, inner sep=0.8pt, fill=black, label={below:{$x_1(t_0)$}}] (A1) {};
		\draw (-8.35, -2.5) node[circle, inner sep=0.8pt, fill=black, label={above:{$x_1(t_1)$}}] (B1) {};
		\draw (-7.50, -0.85) node[circle, inner sep=0.8pt, fill=black, label={left:{$x_1(t_2)$}}] (C1) {};
		\draw (-6.2, -2.5) node[circle, inner sep=0.8pt, fill=black, label={above:{$x_1(t_3)$}}] (D1) {};
		
		\draw (-1.15, -2.5) node[circle, inner sep=0.8pt, fill=black, label={above:{$x_2(t_2')$}}] (A2) {};
		\draw (-2.5, -3.05) node[circle, inner sep=0.8pt, fill=black, label={below:{$x_2(t_1')$}}] (B2) {};
		\draw (-5.1, -3.7) node[circle, inner sep=0.8pt, fill=black, label={below:{$x_2(t_0')$}}] (C2) {};
		\draw (-2.5, -1.35) node[circle, inner sep=0.8pt, fill=black, label={above:{$x_2(t_3')$}}] (D2) {};
		
		\node at (-8.7, 0.5) {$S_1$};
		\node at (-5.0, 0.5) {$S_2$};
		\node at (-1.3, 0.5) {$S_3$};
		\node at (-1.3, -5.5) {$S_4$};
		\node at (-5.0, -5.5) {$S_5$};
		\node at (-8.7, -5.5) {$S_6$};
		
		
		\end{tikzpicture}
		
		\caption{An example of two agents performing in a partitioned workspace.}
		\label{fig: example_01}
	\end{figure}
\end{example}

\subsection{Problem Statement}

We are now ready to define our problem formally as follows:

\begin{problem} \label{problem: basic_prob}
Given $N$ agents that are governed by dynamics as in \eqref{eq: system} and $N$ task specification formulas $\varphi_1, \ldots, \varphi_N$ expressed in MITL over the sets of services $\Sigma_1, \ldots, \Sigma_{{N}}$, respectively, the partition $S = \{S_\ell\}_{\ell \in \mathbb I}$ as in Assumption \ref{assumption: AP_cell_decomposition} and the labeling functions $\Lambda_1, \ldots, \Lambda_N$ and $\mathcal{L}_1, \ldots, \mathcal{L}_N$, as in \eqref{eq:label_lambda}, \eqref{eq:label_lambda2} respectively, assign control laws to the free inputs $v_1, \ldots, v_N$ such that each agent fulfills its individual specification, given the bound $v_{\text{max}}$.
\end{problem}

\begin{remark}
In our preliminary work on the multi-agent controller synthesis framework under MITL specifications \cite{alex_2016_acc}, the multi-agent system was considered to have fully-actuated dynamics. The only constraints on the system were due to the presence of time constrained MITL formulas. In the current framework, we have two types of constraints. Primarily, due to the coupled dynamics of the system, which constrain the motion of each agent, and, secondly, the timed constraints that are inherently imposed from the time bounds of the MITL formulas. Thus, there exist formulas that cannot be satisfied either due to the coupling constraints or the time constraints of the MITL formulas. These constraints, make the procedure of the controller synthesis in the discrete level substantially different and more elaborate than the corresponding multi-agent LTL frameworks in the literature (\cite{guo_2015_reconfiguration, frazzoli_vehicle_routing, belta_regular2, belta_cdc_reduced_communication}).
\end{remark}

\begin{remark}
It should be noted here that, in this work, the couplings and the dependencies between the agents are treated by the dynamics of the form \eqref{eq: system} and not in the discrete level by coupling also in the services of each agent (i.e., $\Sigma_i \cap \Sigma_j \ne \emptyset, \ \forall i, j \in \mathcal I$). Hence, even though the agents do not share atomic propositions, it is the constraint to their motion due to the dynamic couplings with the neighbors that restrict them to fulfill the desired high-level tasks. Treating the coupling through individual atomic propositions in the discrete level as well, constitutes another problem which is far from trivial and a topic of current work.
\end{remark}

\begin{remark}
The motivation for introducing the cell decomposition $S = \{S_\ell\}_{\ell \in \mathbb I}$ in this Section, is that it is required to know which services hold in each part of the workspace. As will be witnessed through the problem solution, this is necessary since the abstraction of the workspace may not be compliant with the initial given cell decomposition, so new partitions and cell decompositions might be required.
\end{remark}

\section{Proposed Solution} \label{sec: solution}

In this section, a systematic solution to Problem~\ref{problem: basic_prob} is introduced. Our overall approach builds on abstracting the system in \eqref{eq: system} into a set of WTSs and further the fact that the timed runs in the $i$-th WTS project onto the trajectories of agent $i$ while preserving the satisfaction of the individual MITL formulas $\varphi_i, i \in \mathcal{I}$. We take the following steps:
\begin{enumerate}
	\item Initially, the boundedness of the agents' relative positions is proved, in order to guarantee boundedness of the coupling terms $-\sum_{j \in \mathcal{N}(i)}^{} (x_i - x_j)$. This property, is required for the derivation of the symbolic models. (Sec. \ref{sec: boundedeness}).
	\item We utilize decentralized abstraction techniques for the multi-agent system, i.e., discretization of both the workspace and the time such that the motion of each agent is modeled by a WTS $\mathcal{T}_i, \ i \in \mathcal{I}$ (Sec. \ref{sec: abstration}).
	\item In view of the definition of WTS, the run of each agent is defined such as to be consistent in view of the coupling constraints with the neighbors. The computation of the product of each individual WTS is thus also required (Sec. \ref{sec: runs consistency}).
	\item A five-step automated procedure for controller synthesis which serves as a solution to Problem \ref{problem: basic_prob} is provided in Sec. \ref{sec: synthesis}.
	\item Finally, the computational complexity of the proposed approach is discussed in Sec. \ref{sec:complexity}.
\end{enumerate}
The next sections provide the proposed solution in detail.

\subsection{Boundedness Analysis} \label{sec: boundedeness}

\begin{theorem} \label{theorem: theorem_1}
	Consider the multi-agent system \eqref{eq: system}. Assume that the network graph is connected (i.e. $\lambda_2(\mathcal{G}) > 0$) and let $v_i, i \in \mathcal{I}$ satisfy $\| v_i(t) \| \leq v_{\text{max}}, \ \forall \ i \in \mathcal{I}, t \geq 0$. Furthermore, let a positive constant $\bar{R} > K_2 v_{\text{max}}$ where $K_2 = \frac{2 \sqrt{N} (N-1) \left\| D(\mathcal{G})^{\tau} \right\|}{\lambda_2^2(\mathcal{G})} > 0$ and where $D(\mathcal{G})$ is the network adjacency matrix. Then, for each initial condition $x_i(0) \in \mathbb{R}^n$, there exists a time $T > 0$ such that $\tilde{x}(t) \in \mathcal{X}, \ \forall t \geq T$, where $\mathcal{X} = \{x \in \mathbb{R}^{Nn} : \|\tilde{x}\| \leq \bar{R}\}$ and $\tilde{x}(t)$ was was defined in \eqref{eq:xtilde}.
\end{theorem}
\begin{proof}
	Consider the following candidate Lyapunov function $V: \mathbb{R}^{Nn} \rightarrow \mathbb{R}$
	\begin{equation}
		V(x) = \frac{1}{2} \sum_{i=1}^{N} \sum_{j \in \mathcal{N}(i)}^{} \| x_i-x_j\|^2 = \|\tilde{x}\|^2 > 0.
	\end{equation}
	The time derivative of $V$ along the trajectories of \eqref{eq: system}, can be computed as
	\begin{align}
		\dot{V}(x) &= \left[ \nabla V(x) \right]^{\tau} \ \dot{x} \nonumber \\
		&=
		\displaystyle \sum_{k=1}^{n} \left\{\frac{\partial V}{\partial x_1^k} \ \dot{x}_1^k  \right\} + \ldots + \sum_{k=1}^{n} \left\{\frac{\partial V}{\partial x_N^k} \ \dot{x}_N^k  \right\} \nonumber \\
		&=
		\sum_{k=1}^{n} \left\{ c\left(\frac{\partial V}{\partial x} ,k\right)^{\tau} \ c(\dot{x},k) \right\} \notag \\
		&= \sum_{k=1}^{n} \left\{ c\left(\frac{\partial V}{\partial x} ,k\right)^{\tau} \left[ -L(\mathcal{G}) \ c(x,k)+c(v, k) \right] \right\}\label{lyap1}
	\end{align}
	where $c(\frac{\partial V}{\partial x} ,k) = \begin{bmatrix} \frac{\partial V}{\partial x_1^k} & \dots & \frac{\partial V}{\partial x_N^k} \end{bmatrix}^{\tau}$. By computing the partial derivative of the Lyapunov function with respect to vector $x_i, \ i \in \mathcal{I}$ we get $\frac{\partial V}{\partial x_i} = \sum_{j \in \mathcal{N}(i)}^{} (x_i-x_j), \ i \in \mathcal{I}$ from which we have that $c\left(\frac{\partial V}{\partial x} ,k\right)^{\tau} = c(x,k)^{\tau} \ L(\mathcal{G}), \ k = 1,...,n$. Thus, by substituting the last in \eqref{lyap1} we get
	\begin{align}
		\dot{V}(x) &= \sum_{k=1}^{n} \left\{ c(x,k)^{\tau} \ L(\mathcal{G}) \ \left[ -L(\mathcal{G}) \ c(x,k)+c(v, k) \right] \right\} \nonumber \\
		&= -\sum_{k=1}^{n} \left\{ c(x,k)^{\tau} \ \left[ L(\mathcal{G})\right]^2 \ c(x,k) \right\} + \nonumber \\ 
		&\qquad \qquad \qquad \sum_{k=1}^{n} \left\{ c(x,k)^{\tau} \ \left[ L(\mathcal{G})\right]^2 c(v, k) \right\} \nonumber \\
		&\leq -\sum_{k=1}^{n} \left\{ c(x,k)^{\tau} \ \left[ L(\mathcal{G})\right]^2 \ c(x,k)) \right\} + \nonumber \\ 
		&\qquad \qquad \qquad \left\| \sum_{k=1}^{n} \left\{ c(x,k)^{\tau} \ L(\mathcal{G}) \ c(v, k) \right\} \right\|. \label{eq: lyap3}
	\end{align}
	
	\noindent For the first term of \eqref{eq: lyap3} we have that
	\begin{align}
		\sum_{k=1}^{n} \left\{ c(x,k)^{\tau} \ L(\mathcal{G})^2 \ c(x,k)) \right\} = \sum_{k=1}^{n}  \left\| L(\mathcal{G}) \ c(x,k) \right\|^2. \notag
	\end{align}
	
	\noindent For the second term of \eqref{eq: lyap3} we have that
	\begin{align}
		&\left\| \sum_{k=1}^{n} \left\{ c(x,k)^{\tau} \ L(\mathcal{G}) \ c(\nu,k) \right\} \right\|  \notag \\
		&\qquad = \left\| \sum_{k=1}^{n} \left\{ c(x,k)^{\tau} \ D(\mathcal{G}) \ D(\mathcal{G})^{\tau} \ c(\nu,k) \right\} \right\| \notag \\
		&\leq \sum_{k=1}^{n} \left\{ \left\| D(\mathcal{G}) ^{\tau} \ c(x,k) \right\| \ \left\| D(\mathcal{G})^{\tau} \right\| \ \left\| c(\nu,k) \right\|\right\} \notag \\
		&\qquad =  \left\| D(\mathcal{G})^{\tau} \right\| \ \sum_{k=1}^{n} \left\{ \left\| c(\tilde{x},k) \right\| \ \left\| c(\nu,k) \right\| \right\}. \label{eq: lyap4}
	\end{align}
	By using the Cauchy-Schwarz inequality in \eqref{eq: lyap4} we get	
	\begin{align}
		&\left\| \sum_{k=1}^{n} \left\{ c(x,k)^{\tau} \ L(\mathcal{G}) \ c(v, k) \right\} \right\| \nonumber \\ 
		&\qquad \leq \left\| D(\mathcal{G})^{\tau} \right\| \ \left( \sum_{k=1}^{n} \left\| c(\tilde{x},k) \right\|^2 \right)^{\frac{1}{2}} \ \left( \sum_{k=1}^{n} \left\| c(v, k) \right\|^2 \right)^{\frac{1}{2}} \notag \\
		&\qquad = \left\| D(\mathcal{G})^{\tau} \right\| \|\tilde{x}\| \|v\| \leq \left\| D(\mathcal{G})^{\tau} \right\| \|\tilde{x}\| \sqrt{N} \|v\|_\infty \notag
	\end{align}
	
	\noindent where $\|v\|_\infty = \text{max} \left\{\|v_i\| : i = 1, \ldots, N\right\} \leq v_{\text{max}}$. Thus, by combining the previous inequalities, \eqref{eq: lyap3} is written
	\begin{equation}
		\dot{V}(x) \leq -\sum_{k=1}^{n}  \left\{ \left\| L(\mathcal{G}) c(x,k) \right\|^2 \right\} + \sqrt{N} \left\| D(\mathcal{G})^{\tau} \right\| \|\tilde{x}\| v_{\text{max}}. \label{eq: lyap8}
	\end{equation}
	In order to proceed the following Lemma is required.
	\begin{lemma} \label{lemma: lemma_1}
		Let $x^\perp$ be the projection of the vector $x \in \mathbb{R}^{Nn}$ to the orthogonal complement of the subspace $H = \{x \in \mathbb{R}^{Nn}: x_1 = \ldots = x_N\}$. Then, the following hold:
		\begin{align}
			\| L(\mathcal{G}) \ c(x, k) \| &\geq \lambda_2 (\mathcal{G}) \ \|c(x^{\perp}, k) \|, \ \forall \ k \in \mathcal{I} \\
			\|x^\perp\| &\geq \frac{1}{\sqrt{2(N-1)}} \|\tilde{x}\|.
		\end{align}
	\end{lemma}
	\begin{proof}
		See the Appendix of \cite{boskos_cdc_connectivity}.
	\end{proof}
	By exploiting Lemma \ref{lemma: lemma_1}, \eqref{eq: lyap8} is written
	\begin{align}
		\dot{V}(x) &\leq -\lambda_2^2(\mathcal{G}) \sum_{k=1}^{n} \left\{\left\| c(x^{\perp},k) \right\|^2 \right\} + \sqrt{N} \left\| D(\mathcal{G})^{\tau} \right\| \|\tilde{x}\| v_{\text{max}} \notag \\
		&= -\lambda_2^2(\mathcal{G}) \ \|x^\perp\|^2  + \sqrt{N} \left\| D(\mathcal{G})^{\tau} \right\| \|\tilde{x}\| v_{\text{max}} \notag \\
		&\leq -\frac{\lambda_2^2(\mathcal{G})}{2(N-1)}\|\tilde{x}\|^2  + \sqrt{N} \left\| D(\mathcal{G})^{\tau} \right\| \|\tilde{x}\| v_{\text{max}} \notag \\
		&\leq -K_1 \|\tilde{x}\| \left(\|\tilde{x}\|-K_2 v_{\text{max}} \right). \label{eq:concl_1}
	\end{align}	
		where $K_1 = \frac{\lambda_2^2(\mathcal{G})}{2(N-1)} > 0$. By using the following implication $\tilde{x} = D^{\tau}(\mathcal{G})x \Rightarrow \|\tilde{x}\| = \|D(\mathcal{G})^\tau x\| \leq \|D(\mathcal{G})^\tau\| \|x\|$, apparently, we have that $0 < V(x)  = \|\tilde{x}\|^2 \leq \|D(\mathcal{G})^\tau\|^2 \|x\|^2 \ \text{and} \ \dot{V}(x) < 0$ when $\|\tilde{x}\| \geq \bar{R} > K_2 v_{\text{max}} $. Thus, there exists a finite time $T > 0$ such that the trajectory will enter the compact set $\mathcal{X} = \{x \in \mathbb{R}^{Nn} : \|\tilde{x}\| \leq \bar{R}\}$ and remain there for all $t \geq T$ with $\bar{R} > K_2 v_{\text{max}}$. This can be extracted from the following.
		
		Let us define the compact set $\Omega = \left\{ x \in \mathbb{R}^{Nn} : K_2 v_{\text{max}} < \bar{R} \leq \|\tilde{x}\| \leq \bar{M} \right\}$, where $\bar{M} = V(x(0)) = \|\tilde{x}(0)\|^2$. Without loss of generality it is assumed that it holds $\bar{M} > \bar{R}$. Let us define the compact sets:
		\begin{align*}
		\mathcal{S}_{1} &= \left\{ x \in \mathbb{R}^{Nn} : \|\tilde{x}\| \leq \bar{M} \right\}, \\
		\mathcal{S}_{2} &= \left\{ x \in \mathbb{R}^{Nn} : \|\tilde{x}\| \leq K_2 v_{\text{max}} \right\}.
		\end{align*}
		From the equivalences $\forall \ x \in S_1 \Leftrightarrow V(x) = \|\tilde{x}\|^2 \leq \bar{M}^2, \forall \ x \in S_2 \Leftrightarrow V(x) = \|x\|^2 \leq K^2_2 v^2_{\text{max}}$, we have that the boundaries $\partial S_1, \partial S_2$ of sets $S_1, S_2$ respectively, are two level sets of the Lyapunov function $V$. By taking the above into consideration we have that $\partial S_2 \subset \partial S_1$. Hence, we get from \eqref{eq:concl_1} that there exist constant $\gamma > 0$ such that: 
		\begin{equation}
		\dot{V}(x) \leq - \gamma < 0, \forall \ x \in \Omega = S_1 \backslash S_2. \label{eq:concl_2}
		\end{equation}			
		Consequently, the trajectory has to enter the interior of the set of $S_2$ in finite time $T > 0$ and remain there for all time $t \geq T$.
\end{proof}

It should be noticed that the relative boundedness of the agents' positions guarantees a global bound on the coupling terms $-\sum_{j \in \mathcal{N}(i)}^{} (x_i - x_j)$, as defined in \eqref{eq: system}. This bound will be later exploited in order to capture the behavior of the system in $\mathcal{X} = \{x \in \mathbb{R}^{Nn} : \|\tilde{x}\| \leq \bar{R}\}$, by a discrete state WTS.

\subsection{Abstraction} \label{sec: abstration}

In this section we provide the abstraction technique that is adopted in order to capture the dynamics of each agent into Transition Systems. We work completely in discrete level, which is necessary in order to solve the Problem \ref{problem: basic_prob}.

Firstly, some additional notation is introduced. Given an index set $\mathbb{I}$ and an agent $i \in \mathcal{I}$ with neighbors $j_1, \ldots, j_{N_i}$, the mappings $\text{pr}_i : \mathbb{I}^N \rightarrow \mathbb{I}^{N_i+1}, \bar{\text{pr}}_i : \mathbb{I}^N \rightarrow \mathbb{I}$ are defined, where $\mathbb{I}^N = \underbrace{\mathbb{I} \times \ldots \times \mathbb{I}}_{N-products}$. The first one assigns to each $N$-tuple ${\bf{l}}= (l_1, \ldots, l_N) \in \mathbb{I}^N$ the $N_i+1$ tuple ${\bf{l}}_i = (l_i, l_{j_1}, \ldots, l_{j_{N_i}}) \in \mathbb{I}^{N_i+1}$ which denotes the indices of the cells where the agent $i$ and its neighbors belong. The second one assigns to each $N$-tuple ${\bf}{l} = (l_1, \ldots, l_N) \in \mathbb{I}^N$ the {position} $l_i \in \mathbb{I}$ of the agent $i$, i.e., the cell that the agent $i$ occupies at the moment.

Consider a particular configuration $\bar{S} = \{\bar{S}_{l}\}_{l \in \bar{\mathbb{I}}}$, where agent $i$ occupies the cell $\bar{S}_{l_i}$. We denote here with $\bar{S}$ the cell decomposition which is the outcome of the abstraction technique that is adopted for the problem solution that will be presented in this Section. This is not necessarily the same the cell decomposition $S$ from Assumption \ref{assumption: AP_cell_decomposition} and Problem \ref{problem: basic_prob}. Let $\delta t$ be a time step. Through the aforementioned space and time discretization $\bar{S}-\delta t$ we aim to capture the reachability properties of the continuous system \eqref{eq: system}, in order to create a WTS of each agent. The WTS will later on serve in the synthesis of plans that fulfill the high-level specifications and that map onto the desired free inputs $v_i, i \in \mathcal I$. 

We proceed by describing the abstraction procedure. If there exists a free input for each state in $\bar{S}_{l_i}$ that navigates the agent $i$ into the cell $\bar{S}_{l_i'}$ precisely in time $\delta t$, regardless of the locations of the agent $i$'s neighbors within their current cells, then a transition from $l_i$ to $\l_i'$ is enabled in the WTS. This forms the well-possessedness of transitions which will be explained hereafter. A mathematical derivation can be found in \cite{boskos_cdc_2015}.

\begin{figure}[ht!] 
	\begin{center}
		\begin{tikzpicture} [scale=.80]
		
		\draw[color=gray,thick] (0,0) -- (4,0);
		\draw[color=gray,thick] (0,1) -- (4,1);
		\draw[color=gray,thick] (0,2) -- (4,2);
		\draw[color=gray,thick] (0,3) -- (4,3);
		
		\draw[color=gray,thick] (0,0) -- (0,3);
		\draw[color=gray,thick] (1,0) -- (1,3);
		\draw[color=gray,thick] (2,0) -- (2,3);
		\draw[color=gray,thick] (3,0) -- (3,3);
		\draw[color=gray,thick] (4,0) -- (4,3);
		
		\draw[color=blue,very thick] (0,0) -- (1,0) -- (1,1) -- (0,1) -- (0,0);
		\draw[color=green,very thick] (2,1) -- (3,1) -- (3,2) -- (2,2) -- (2,1);
		\draw[color=green,very thick] (3,2) -- (4,2) -- (4,3) -- (3,3) -- (3,2);
		\draw[color=cyan,very thick] (0,2) -- (1,2) -- (1,3) -- (0,3) -- (0,2);
		
		
		\fill[yellow] (0.3,2.5) circle (0.3cm);
		\fill[yellow] (0.4,2.7) circle (0.3cm);
		\fill[yellow] (0.5,2.6) circle (0.3cm);
		
		\draw[black, dashed] (0.3,2.5) circle (0.3cm);
		\draw[black, dashed] (0.4,2.7) circle (0.3cm);
		\draw[black, dashed] (0.5,2.6) circle (0.3cm);

		\draw [color=red,thick,->,>=stealth'](0.2,0.3) .. controls (0.3,1.5) .. (0.3,2.5);
		\fill[black] (0.2,0.3) circle (1.5pt);
		
		\draw [color=red,thick,->,>=stealth'](0.5,0.8) .. controls (0.5,1.5) .. (0.4,2.7);
		\fill[black] (0.5,0.8) circle (1.5pt);
		
		\draw [color=red,thick,->,>=stealth'](0.7,0.6) .. controls (0.7,1.5) .. (0.5,2.6);
		\fill[black] (0.7,0.6) circle (1.5pt);

		\coordinate [label=left:$S_{l_i}$] (A) at (0,0.5);
		\coordinate [label=above left:$x_{i}$] (A) at (0.85,0.15);
		\coordinate [label=right:$S_{l_{j_1}}$] (A) at (3,1.5);
		\fill[black] (2.3,1.4) circle (1.5pt) node[right]{$x_{j_{1}}$};
		\coordinate [label=right:$S_{l_{j_2}}$] (A) at (4,2.5);
		\fill[black] (3.4,2.8) circle (1.5pt) node[below]{$x_{j_{2}}$};
		\coordinate [label=left:$S_{l_i'}$] (A) at (0,2.5);
		
		\draw[color=gray,thick] (5,0) -- (9,0);
		\draw[color=gray,thick] (5,1) -- (9,1);
		\draw[color=gray,thick] (5,2) -- (9,2);
		\draw[color=gray,thick] (5,3) -- (9,3);
		
		\draw[color=gray,thick] (5,0) -- (5,3);
		\draw[color=gray,thick] (6,0) -- (6,3);
		\draw[color=gray,thick] (7,0) -- (7,3);
		\draw[color=gray,thick] (8,0) -- (8,3);
		\draw[color=gray,thick] (9,0) -- (9,3);
		
		\draw[color=blue,very thick] (5,0) -- (6,0) -- (6,1) -- (5,1) -- (5,0);
		\draw[color=green,very thick] (7,1) -- (8,1) -- (8,2) -- (7,2) -- (7,1);
		\draw[color=green,very thick] (8,2) -- (9,2) -- (9,3) -- (8,3) -- (8,2);
		\draw[color=cyan,dashed,very thick] (5,2) -- (6,2) -- (6,3) -- (5,3) -- (5,2);
		\draw[color=cyan,dashed,very thick] (6,2) -- (7,2) -- (7,3) -- (6,3) -- (6,2);
		\draw[color=cyan,dashed,very thick] (6,1) -- (7,1) -- (7,2) -- (6,2) -- (6,1);

		\fill[yellow] (5.3,2.5) circle (0.3cm);
		\fill[yellow] (6.4,2.7) circle (0.3cm);
		\fill[yellow] (6.5,1.7) circle (0.3cm);
		
		\draw[black, dashed] (5.3,2.5) circle (0.3cm);
		\draw[black, dashed] (6.4,2.7) circle (0.3cm);
		\draw[black, dashed] (6.5,1.7) circle (0.3cm);
		
		\draw [color=red,thick,->,>=stealth'](5.2,0.3) .. controls (5.3,1.5) .. (5.3,2.5);
		\fill[black] (5.2,0.3) circle (1.5pt);
		
		\draw [color=red,thick,->,>=stealth'](5.5,0.8) .. controls (5.5,1.5) .. (6.4,2.7);
		\fill[black] (5.5,0.8) circle (1.5pt);
		
		\draw [color=red,thick,->,>=stealth'](5.7,0.6) .. controls (5.7,1) .. (6.5,1.7);
		\fill[black] (5.7,0.6) circle (1.5pt);

		\coordinate [label=below:System (i)] (A) at (2,-0.5);
		\coordinate [label=below:System (ii)] (A) at (7,-0.5);
		
		\coordinate [label=left:$S_{l_i}$] (A) at (5,0.5);
		\coordinate [label=above left:$x_{i}$] (A) at (5.85,0.15);
		\coordinate [label=right:$S_{l_{j_1}}$] (A) at (8,1.5);
		\fill[black] (7.3,1.4) circle (1.5pt) node[right]{$x_{j_{1}}$};
		\coordinate [label=right:$S_{l_{j_2}}$] (A) at (9,2.5);
		\fill[black] (8.4,2.8) circle (1.5pt) node[below]{$x_{j_{2}}$};
		
		\draw[dashed,->,>=stealth] (0.4,2.7) -- (1.5,3.5);
		\coordinate [label=right:$x_{i}(\delta t)$] (A) at (1.5,3.5);
		
		\draw[dashed,->,>=stealth] (5.3,2.5) -- (6.5,3.5);
		\coordinate [label=right:$x_{i}(\delta t)$] (A) at (6.5,3.5);
		
		\end{tikzpicture}
		\caption{Illustration of a space-time discretization which is well posed for system (i) but non-well posed for system (ii).}
		\label{fig: well_possed_abstraction}
	\end{center}
\end{figure}
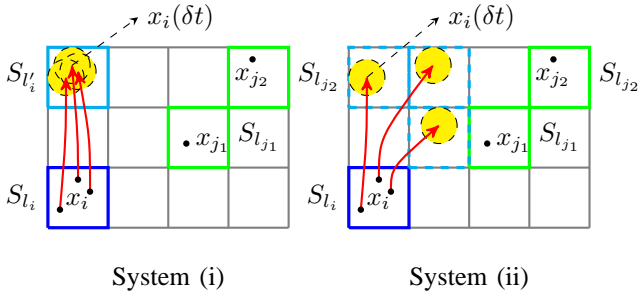

We next illustrate the concept of a well-posed abstraction, namely, a discretization which generates for each agent a Transition System in accordance with the discussion above and the Def. \ref{def: WTS}. Consider a cell decomposition $\bar{S} = \{\bar{S}_{l}\}_{l \in \bar{\mathbb{I}} = \{1, \ldots, 12\}}$ as depicted in Fig. \ref{fig: well_possed_abstraction} and a time step $\delta t$. The tails and the tips of the arrows in the figure depict the initial cell and the endpoints of agent's $i$ trajectories at time $\delta t$ respectively. In both cases in the figure we focus on agent $i$ and consider the same cell configuration for $i$ and its neighbors. However, different dynamics are considered for Cases (i) and (ii). In Case (i), it can be observed that for the three distinct initial positions in cell $\bar{S}_{l_i}$, it is possible to drive agent $i$ to cell $\bar{S}_{l_i'}$ at time $\delta t$. We assume that this is possible for all initial conditions in this cell and irrespectively of the initial conditions of $i$'s neighbors in their cells and the inputs they choose. It is also assumed that this property holds for all possible cell configurations of $i$ and for all the agents of the system. Thus we have a well-posed discretization for system (i). On the other hand, for the same cell configuration and system (ii), the following can be observed. For three distinct initial conditions of $i$ the corresponding reachable sets at $\delta t$, which are enclosed in the dashed circles, lie in different cells. Thus, it is not possible given this cell configuration of $i$ to find a cell in the decomposition which is reachable from every point in the initial cell and we conclude that discretization is not well-posed for system (ii).

We present at this point the sufficient conditions that relate the dynamics of the multi-agent system \eqref{eq: system}, the time step $\delta t$ and the diameter $d_{\text{max}} = \text{sup} \{ \|x - y\| : x,y \in \bar{S}_l, l \in \mathbb{I} \}$ of the cell decomposition $\bar{S}$, and guarantee the existence of the aforementioned well-posed transitions from each cell. Based on \cite{boskos_cdc_2015} (Section III, inequality (3), Section IV, inequalities (28, 29)), the sufficient conditions that the dynamics of a general class of system in the form
	\begin{equation} \label{eq:boskos_dynamics}
	\dot{x}_i = f_i(x_i, \mathbf{x}_j)+v_i, i \in \mathcal{I}
	\end{equation}
	where $\mathbf{x}_j = (x_{j_1}, \ldots, x_{j_{N_i}}) \in \mathbb{R}^{N_i n}$, should fulfill in order to have well-posed abstractions are the following:

\noindent \textbf{(C1)} There exists $M > v_{\text{max}} > 0$ such that $\|f_i(x_i, \mathbf{x}_j) \| \leq M, \ \forall i \in \mathcal{I}, \forall \ x \in \mathbb{R}^{Nn} : \text{pr}_i(x) = (x_i, \mathbf{x}_j) \ \text{and} \  \tilde{x} \in \mathcal{X}$, by applying the projection operator $\text{pr}_i$ for $\mathbb{I} = \mathbb{R}^n$.

\noindent \textbf{(C2)} There exists a Lipschitz constant $L_1 > 0$ such that
\begin{align}
&\| f_i(x_i, \mathbf{x}_j) - f_i(x_i, \mathbf{y}_j) \| \leq L_1 \|(x_i, \mathbf{x}_j) - (x_i, \mathbf{y}_j) \|, \notag \\  
&\qquad \qquad \forall \ i \in \mathcal{I}, x_i, y_i \in \mathbb{R}^n, \mathbf{x}_j, \mathbf{y}_j \in \mathbb{R}^{N_in}.
\end{align}

\noindent \textbf{(C3)} There exists a Lipschitz constant $L_2 > 0$ such that
\begin{align}
&\| f_i(x_i, \mathbf{x}_j) - f_i(y_i, \mathbf{x}_j) \| \leq L_2 \|(x_i, \mathbf{x}_j) - (y_i, \mathbf{x}_j) \|, \notag \\ 
&\qquad \qquad \forall \ i \in \mathcal{I}, x_i, y_i \in \mathbb{R}^n, \mathbf{x}_j, \mathbf{y}_j \in \mathbb{R}^{N_i n}.
\end{align}

\noindent From \eqref{eq: system} and \eqref{eq:boskos_dynamics} we get $f_i(x_i, \mathbf{x}_j) = -\sum_{j \in \mathcal{N}(i)}^{}(x_i-x_j)$. By checking all the conditions one by one for $f_i(x_i, \mathbf{x}_j)$ as in \eqref{eq: system}, we have: \\

\noindent \textbf{(C1)} For every $i \in \mathcal{I}, \forall \ x \in \mathbb{R}^{Nn} : \tilde{x} \in \mathcal{X}$ and $\text{pr}_i(x) = (x_i, \mathbf{x}_j)$ we have that
\begin{align}
\| f_i(x_i, \mathbf{x}_j) \| &= \left\| -\sum_{j \in \mathcal{N}(i)}^{}(x_i-x_j) \right\|  \leq \sum_{j \in \mathcal{N}(i)}^{} \| x_i-x_j \| \notag \\ 
&\leq  \sum_{ (i,j) \in \mathcal{E} }^{} \| x_i-x_j \| = \Delta x \leq \bar{R}.
\end{align}
Thus, $M = \bar{R}$. We have also that $\|D(\mathcal{G})^\tau\| = \sqrt{\lambda_{\text{max}} ( D(\mathcal{G}) D(\mathcal{G}) ^\tau)} = \sqrt{\lambda_{\text{max}}(\mathcal{G})}$ and $\lambda_2(\mathcal{G}) \leq \frac{N}{N-1} \min \{N_i : i \in \mathcal{I} \}$ from \cite{fiedler1973algebraic}.  For $N > 2$ it holds that $\lambda_2(\mathcal{G}) < N$. From Theorem \eqref{theorem: theorem_1} we have that $\bar{R} > K_2 v_{\text{max}} \Leftrightarrow M > K_2 v_{\text{max}}$. It holds that $M > v_{\text{max}}$ since
\begin{align}
K_2 &= \frac{2 \sqrt{N} (N-1) \left\| D(\mathcal{G})^{\tau} \right\|}{\lambda_2^2(\mathcal{G})} = \frac{2 \sqrt{N} (N-1) \sqrt{\lambda_{\text{max}}(\mathcal{G})}}{\sqrt{\lambda_2^3(\mathcal{G})} \sqrt{\lambda_2(\mathcal{G})}} \notag \\
&\geq \frac{2 \sqrt{N} (N-1)}{\sqrt{N^3}} \sqrt{\frac{\lambda_{\text{max}}(\mathcal{G})}{\lambda_2(\mathcal{G})}} \geq \frac{2 \sqrt{N} (N-1)}{\sqrt{N^3}} > 1. \notag
\end{align}

\noindent \textbf{(C2)} We have that
\begin{align}
&\| f_i(x_i, \mathbf{x}_j) - f_i(x_i, \mathbf{y}_j) \| \notag \\ 
&\qquad = \Big\| -\sum_{j \in \mathcal{N}(i)}^{}(x_i-x_j) + \sum_{j \in \mathcal{N}(i)}^{}(x_i-y_j) \Big\| \notag \\
&\qquad \leq \sqrt{N_i} \ \|(x_i, \mathbf{x}_j) - (x_i, \mathbf{y}_j) \| \notag \\
&\qquad \leq  \text{max}\{\sqrt{N_i} : i = 1, \ldots, N\} \ \|(x_i, \mathbf{x}_j) - (x_i, \mathbf{y}_j) \|. \notag
\end{align}
Thus, the condition \textbf{(C2)} holds and the Lipschitz constant is $L_1 = \text{max}\{\sqrt{N_i} : i = 1, \ldots, N\} >0$, where the inequality $\left( \sum_{i = 1}^{\rho} \alpha_i \right)^2 \leq \rho \ \left( \sum_{i=1}^\rho \alpha_i^2 \right)$ is used. \\

\noindent \textbf{(C3)} By using the same methodology with the proof of \textbf{(C2)} we conclude that $L_2 = \text{max} \{N_i : i = 1, \ldots, N\} > 0$.

Based on the sufficient condition for well posed abstractions in \cite{boskos_cdc_2015}, the acceptable values of $d_{\max}$ and $\delta t$ are given as
\begin{align}
d_{\max} & \in \left(0,\frac{(1-\lambda)^2 v_{\max}^{2}}{4ML}\right] \label{dmax} \\
\delta t & \in \left[\frac{(1-\lambda)v_{\max}-\sqrt{(1-\lambda)^2 v_{\max}^{2}-4MLd_{\max}}}{2ML}, \right. \notag \\
&\left. \frac{(1-\lambda)v_{\max}+\sqrt{(1-\lambda)^2 v_{\max}^{2}-4MLd_{\max}}}{2ML}\right] \label{deltat}
\end{align}

\noindent where the parameter $\lambda$ stands for reachability purposes and $L=\max\{3L_{2}+4L_{1}\sqrt{N_i},i\in\mathcal{I}\}$ with the dynamics bound $M$ and the Lipschitz constants $L_1$, $L_2$ as previously deduced.

\begin{remark}
	Notice that when $d_{\max}$ in \eqref{dmax} is chosen sufficiently small, it is also possible due to the lower bound on the acceptable $\delta t$ in \eqref{deltat} to select a correspondingly small value of the sampling time and capture with higher accuracy the properties of the continuous trajectories. However, this will result in a finer discretization and increase the complexity of the symbolic models.
\end{remark}

\begin{remark} \label{remark:d_max_remark}
	Assume that a cell-decomposition of diameter $d_{\text{max}}$ and a time step $\delta t$ which guarantee well-posed transitions, namely, which satisfy \eqref{dmax} and \eqref{deltat}, have been chosen. It is also possible to chose any other cell-decomposition with diameter $\hat{d}_{\text{max}} \leq d_{\text{max}}$ since, by \eqref{deltat}, the range of acceptable $\delta t$ increases.
\end{remark}

We showed that the dynamics of the system \eqref{eq: system} satisfy all the sufficient conditions (C1)-(C3), thus we have a well-posed space-time discretization $\bar{S}-\delta t$. Recall now Assumption \ref{assumption: AP_cell_decomposition}. It remains to establish the compliance of the cell decomposition $S = \{S_\ell\}_{\ell \in \mathbb{I}}$, which is given in the statement of Problem \ref{problem: basic_prob}, with the cell decomposition $\bar{S} = \{\bar{S}_l\}_{l \in \bar{\mathbb{I}}}$, which is the outcome of the abstraction. By the term of compliance, we mean that:
	\begin{equation} \label{eq:cell_decomposition_compliance}
	\bar{S}_l \cap S_\ell \in S \cup \{\emptyset\}, \ \text{for each} \ \bar{S}_l \in \bar{S}, S_\ell \in S \ \text{and} \ l \in \bar{\mathbb{I}}, \ell \in \mathbb{I}.
	\end{equation} 
	In order to address this problem, we define:
	\begin{equation} \label{eq:final_decomposition}
	\hat{S} = \{ \hat{S}_{\hat{l}} \}_{\hat{l} \in \hat{\mathbb{I}}} = \{ \bar{S}_l \cap S_\ell : \ l \in \bar{\mathbb{I}}, \ell \in \mathbb{I} \} \backslash \{ \emptyset \}
	\end{equation}
	which forms a cell decomposition which is compliant with the cell decomposition $S$ from Problem \ref{problem: basic_prob} and serves as the abstraction solution of this problem. This can be deducted as follows: By taking all the combinations of intersections $\bar{S}_l \cap S_\ell, \forall \ l \in \bar{\mathbb{I}}, \forall \ \ell \in \mathbb{I}$ and enumerating them by indexes of the set $\hat{\mathbb{I}}$, the cells $\{ \hat{S}_{\hat{l}} \}_{\hat{l} \in \hat{\mathbb{I}}}$ are constructed for which the following holds: $\forall \ \ell \in \mathbb{I}, \exists \ l \in \bar{\mathbb{I}}$ such that $\hat{S}_{\hat{l}} = \bar{S}_l \cap S_\ell = \emptyset$ and $\text{int}(\hat{S}_{\hat{\ell}}) \cap \text{int}(\hat{S}_{\hat{\ell}'}) \neq \emptyset$ for all $\hat{\ell}' \in \hat{\mathbb{I}} \backslash \{\hat{\ell}\}$. After all the intersections we have $\cup_{l \in \hat{\mathbb{I}}} \hat{S}_{\hat{l}} = \mathcal{X}$. The diameter of the cell decomposition $\hat{S} =  \{ \hat{S}_{\hat{l}} \}_{\hat{l} \in \hat{\mathbb{I}}}$ is defined as $\hat{d}_{\text{max}} = \text{sup} \{ \|x - y\| : x,y \in \hat{S}_{\hat{l}}, \hat{l} \in \hat{\mathbb{I}} \} \leq d_{\text{max}}$. Hence, according to the discussion above, we have a well-posed abstraction. See Example \ref{ex: example_03} for an illustration of these deviations. 

For the solution to Problem \ref{problem: basic_prob}, the WTS which corresponds to the cell configuration $\hat{S}$, the diameter $\hat{d}_{\text{max}}$ and the time step $\delta t$ will be exploited. Thus, the WTS of each agent is defined as follows:
\begin{definition} \label{def: indiv_WTS}
	The motion of each agent $i \in \mathcal{I}$ in the workspace is modeled by a WTS $\mathcal{T}_i = (S_i, S_i^{\text{init}}, Act_i, \longrightarrow_i,$ $ d_i, AP_i, \hat L_i)$ where
	\begin{itemize}
		\item $S_i = \hat{\mathbb{I}}$, the set of states of each agent is the set of indices of the cell decomposition.
		\item $S_0^{\text{init}} \subseteq S_i$, is a set of initial states.
		\item $Act_i = {\hat{\mathbb{I}}}^{N_i+1}$, the set of actions representing where agent $i$ and its neighbors are located.
		\item For a pair $(l_i, {\bf{l}_i}, l'_i)$ we have that $(l_i, {\bf{l}_i}, l'_i) \in \longrightarrow_i$ iff $l_i \overset{\bf{l}_i}{\longrightarrow_i} l_i'$ is well-posed for each $l_i, l'_i \in S_i$ and ${\bf{l}_i} = (l_i, l_{j_1}, \ldots, l_{j_{N_i}}) \in Act_i$.
		\item $d_i: \longrightarrow_i \rightarrow \mathbb{T}$, is a map that assigns a positive weight (duration) to each transition. The duration of each transition is exactly equal to $\delta t > 0$.
		\item $\AP_i = \Sigma_i$, is the set of atomic propositions which are inherent properties of the workspace.
		\item $\hat L_i: S_i \rightarrow 2^{AP_i}$, is the labeling function that maps the every state $s \in S_i$ into the services that can be provided in this state.
	\end{itemize}
\end{definition}

The aforementioned Definition is crucial since from the dynamical system \eqref{eq: system} we created through the abstraction procedure individual WTSs of each agent $i \in \mathcal{I}$ which capture the motion of each agent and let us work completely in discrete level in order to design the controllers that satisfy the Problem \ref{problem: basic_prob}.

Every WTS $\mathcal{T}_i, i \in \mathcal{I}$ generates timed runs and timed words of the form $r_i^t = (r_i(0), \tau_i(0))(r_i(1), \tau_i(1))(r_i(2), \tau_i(2)) \ldots, w_i^t = (L_i(r_i(0)), \tau_i(0))(L_i(r_i(1)), \tau_i(1))(L_i(r_i(2)), \tau_i(2)) \ldots$ respectively, over the set $2^{AP_i}$ according to Def. \ref{run_of_WTS} with $\tau_i(j) = j \cdot \delta t, \forall \ j \ge 0$. It is necessary now to provide the relation between the time words that are generated by the WTSs $\mathcal{T}_i, i \in \mathcal{I}$ with the time service words produced by the trajectories $x_i(t), i \in \mathcal{I}, t \ge 0$.

\begin{remark} \label{lemma:compliant_WTS_runs_with_trajectories}
By construction, each time word produced by the WTS $\mathcal{T}_i$ is a service time word associated with the trajectory $x_i(t)$ of the system \eqref{eq: system}. Hence, if we find a timed word of $\mathcal{T}_i$ satisfying a formula $\varphi_i$ given in MITL, we also found for each agent $i$ a desired timed word of the original system, and hence trajectories $x_i(t)$ that are solution to the Problem \eqref{problem: basic_prob}. (i.e., the produced timed words of $\mathcal{T}_i$ are compliant with the service time words of the trajectories $x_i(t)$.)
\end{remark}

\begin{example} \label{ex: example_03}
	Assume that $S = \{S_\ell\}_{\ell \in \{1,\ldots,6\}}$ as given in Example \ref{ex: example_0} depicted in Fig. \ref{ex: example_03} by red rectangles, is the cell decomposition of Problem \ref{problem: basic_prob}. Let also $\bar{S} = \{\bar{S}_l\}_{l \in \bar{\mathbb{I}} = \{1, \ldots, 6\}}$ depicted in Fig. \ref{fig:example_04} with light blue cells, be a cell decomposition which serves as potential solution of this Problem satisfying all the abstraction properties that have been mentioned in this Section. It can be observed that the two cell decompositions are not compliant according to \eqref{eq:cell_decomposition_compliance}. However, by using \eqref{eq:final_decomposition}, a new cell decomposition $\hat{S} =  \{ \hat{S}_{\hat{l}} \}_{\hat{l} \in \hat{\mathbb{I}} = \{1, \ldots, 15\}}$ (depicted in Fig. \ref{fig:example_04}), that is compliant $S$, can be obtained and forms the final cell decomposition solution. Let also $d_{\text{max}}, \hat{d}_\text{max}$ be the diameters of the cell decompositions $S, \hat{S}$ respectively. It holds that $\hat{d}_\text{max} \leq d_{\text{max}}$ which is in accordance with the Remark \ref{remark:d_max_remark}.
	\begin{figure}[ht!]
		\centering
		\begin{tikzpicture}[scale = 0.65]
		
		\draw[step=2.5, line width=.04cm, color = red, draw opacity = 0.9] (-2.5, -5.0) grid (0,0);
		\draw[line width=.04cm, color = red, draw opacity = 0.9] (-7.5,0.0) -- (-2.5,0.0);
		\draw[line width=.04cm, color = red, draw opacity = 0.9] (-7.5,-2.5) -- (-2.5,-2.5);
		\draw[line width=.04cm, color = red, draw opacity = 0.9] (-7.5,-5.0) -- (-2.5,-5.0);
		\draw[step=2.5, line width=.04cm, color = red, draw opacity = 0.9] (-10.0, -5.0) grid (-7.5,0);
		
		\draw[line width=.04cm, color = blue, draw opacity = 0.3] (0,0) rectangle (-3.5,-3.5);
		\draw[line width=.04cm, color = blue, draw opacity = 0.3] (0,-3.5) rectangle (-3.5,-5.0);
		\draw[line width=.04cm, color = blue, draw opacity = 0.3] (-3.5,-3.5) rectangle (-6.5,-5);
		\draw[line width=.04cm, color = blue, draw opacity = 0.3] (-3.5,0) rectangle (-6.5,-3.5);
		\draw[line width=.04cm, color = blue, draw opacity = 0.3] (-6.5,0) rectangle (-10.0,-3.5);
		\draw[line width=.04cm, color = blue, draw opacity = 0.3] (-6.5,-3.5) rectangle (-10.0,-5.0);
		
		\node[color = red] at (-8.7, -1.2) {$S_1$};
		\node[color = red] at (-5.0, -0.9) {$S_2$};
		\node[color = red] at (-1.3, -1.2) {$S_3$};
		\node[color = red] at (-1.3, -3.8) {$S_4$};
		\node[color = red] at (-5.0, -3.8) {$S_5$};
		\node[color = red] at (-8.7, -3.8) {$S_6$};
		
		\draw [-latex, color = blue, draw opacity = 0.3, thick, shorten >= 1.00cm] (-1.8, -1.8) -- (1.5, -1.8);
		\node[color = blue, draw opacity = 0.3] at (1.0, -1.75) {$\bar{S}_3$};
		
		\draw [-latex, color = blue, draw opacity = 0.3, thick, shorten >= 1.00cm] (-1.8, -4.2) -- (1.5, -5.2);
		\node[color = blue, draw opacity = 0.3] at (1.1, -4.9) {$\bar{S}_4$};
		
		\draw [-latex, color = blue, draw opacity = 0.3, thick, shorten >= 1.00cm] (-4.8, -1.8) -- (-3.6, 1.3);
		\node[color = blue, draw opacity = 0.3] at (-3.6, 0.8) {$\bar{S}_2$};
		
		\draw [-latex, color = blue, draw opacity = 0.3, thick, shorten >= 1.00cm] (-4.8, -4.3) -- (-4.8, -6.5);
		\node[color = blue, draw opacity = 0.3] at (-4.8, -5.9) {$\bar{S}_5$};
		
		\draw [-latex, color = blue, draw opacity = 0.3, thick, shorten >= 1.00cm] (-8.5, -4.3) -- (-8.5, -6.5);
		\node[color = blue, draw opacity = 0.3] at (-8.5, -5.9) {$\bar{S}_5$};
		
		\draw [-latex, color = blue, draw opacity = 0.3, thick, shorten >= 1.00cm] (-8.5, -1.7) -- (-11.5, -1.7);
		\node[color = blue, draw opacity = 0.3] at (-11, -1.7) {$\bar{S}_1$};
		
		\draw[color=black,very thick,<->] (-2.5,0) -- (-7.5,-2.5);
		\coordinate [label=below right:$\bf{d_{\max}}$] (A) at (-6.2,-1.8);
		
		\end{tikzpicture}
		\caption{An example with a given cell decomposition $S = \{S_l\}_{l \in \{1,\ldots,6\}}$ and a non-compliant solution $\bar{S} = \{\bar{S}_l\}_{l \in \bar{\mathbb{I}} = \{1, \ldots, 6\}}$.}
		\label{fig: example_03}
	\end{figure}
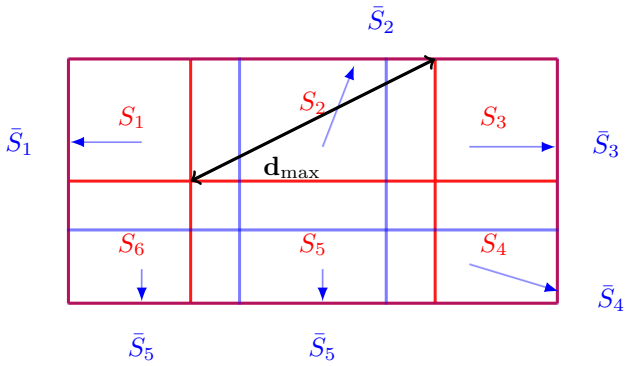
	
	\begin{figure}[ht!]
		\centering
		\begin{tikzpicture}[scale = 0.65]
		
		\draw[step=2.5, line width=.04cm, color = black, draw opacity = 0.9] (-2.5, -5.0) grid (0,0);
		\draw[line width=.04cm, color = black, draw opacity = 0.9] (0,0) rectangle (-3.5,-3.5);
		\draw[line width=.04cm, color = black, draw opacity = 0.9] (-7.5,0.0) -- (-2.5,0.0);
		\draw[line width=.04cm, color = black, draw opacity = 0.9] (-7.5,-2.5) -- (-2.5,-2.5);
		\draw[line width=.04cm, color = black, draw opacity = 0.9] (-7.5,-5.0) -- (-2.5,-5.0);
		\draw[step=2.5, line width=.04cm, color = black, draw opacity = 0.9] (-10.0, -5.0) grid (-7.5,0);
		
		\draw[line width=.04cm, color = black, draw opacity = 0.9] (0,-3.5) rectangle (-3.5,-5.0);
		\draw[line width=.04cm, color = black, draw opacity = 0.9] (-3.5,-3.5) rectangle (-6.5,-5);
		\draw[line width=.04cm, color = black, draw opacity = 0.9] (-3.5,0) rectangle (-6.5,-3.5);
		\draw[line width=.04cm, color = black, draw opacity = 0.9] (-6.5,0) rectangle (-10.0,-3.5);
		\draw[line width=.04cm, color = black, draw opacity = 0.9] (-6.5,-3.5) rectangle (-10.0,-5.0);
		
		\draw[fill = green!5] (0,0) rectangle (-2.5,-2.5);
		\draw[fill = green!5] (0,-2.5) rectangle (-2.5,-3.5);
		\draw[fill = green!5] (0,-3.5) rectangle (-2.5,-3.5);
		\draw[fill = green!5] (0,-3.5) rectangle (-2.5,-5.0);
		\draw[fill = green!5] (-2.5,0) rectangle (-3.5,-2.5);
		\draw[fill = green!5] (-2.5,-2.5) rectangle (-3.5,-3.5);
		\draw[fill = green!5] (-2.5,0) rectangle (-3.5,-2.5);
		\draw[fill = green!5] (-2.5,-3.5) rectangle (-3.5,-5.0);
		\draw[fill = green!5] (-3.5,0.0) rectangle (-6.5,-2.5);
		\draw[fill = green!5] (-6.5,0.0) rectangle (-7.5,-2.5);
		\draw[fill = green!5] (-7.5,0.0) rectangle (-7.5,-2.5);
		\draw[fill = green!5] (-7.5,0.0) rectangle (-10.0,-2.5);
		\draw[fill = green!5] (-3.5,-2.5) rectangle (-6.5,-3.5);
		\draw[fill = green!5] (-6.5,-2.5) rectangle (-7.5,-3.5);
		\draw[fill = green!5] (-7.5,-2.5) rectangle (-10.0,-3.5);
		\draw[fill = green!5] (-3.5,-3.5) rectangle (-6.5,-5.0);
		\draw[fill = green!5] (-6.5,-3.5) rectangle (-7.5,-5.0);
		\draw[fill = green!5] (-7.5,-3.5) rectangle (-10.0,-5.0);
		
		\node[color = black] at (-8.9, -0.8) {$\hat{S}_1$};
		\node[color = black] at (-7.0, -1.2) {$\hat{S}_2$};
		\node[color = black] at (-5.0, -1.2) {$\hat{S}_3$};
		\node[color = black] at (-3.0, -1.2) {$\hat{S}_4$};
		\node[color = black] at (-1.3, -1.2) {$\hat{S}_5$};
		\node[color = black] at (-1.3, -2.95) {$\hat{S}_6$};
		\node[color = black] at (-3.0, -2.95) {$\hat{S}_7$};
		\node[color = black] at (-5.0, -2.95) {$\hat{S}_8$};
		\node[color = black] at (-7.0, -2.95) {$\hat{S}_{9}$};
		\node[color = black] at (-8.7, -2.95) {$\hat{S}_{10}$};
		\node[color = black] at (-8.7, -4.20) {$\hat{S}_{11}$};
		\node[color = black] at (-7.0, -4.20) {$\hat{S}_{12}$};
		\node[color = black] at (-5.0, -4.20) {$\hat{S}_{13}$};
		\node[color = black] at (-3.0, -4.20) {$\hat{S}_{14}$};
		\node[color = black] at (-1.3, -4.20) {$\hat{S}_{15}$};
		
		\draw[color=red,very thick,<->] (-7.5,0) -- (-10.0,-2.5);
		\node[color = red] at (-9.0, -2.10) {$\hat{d}_{\text{max}}$};
		
		\end{tikzpicture}
		\caption{The resulting compliant cell decomposition $\hat{S} =  \{ \hat{S}_{\hat{l}} \}_{\hat{l} \in \hat{\mathbb{I}} = \{1, \ldots, 15\}}$ of Example \ref{ex: example_03}.}
		\label{fig:example_04}
	\end{figure}
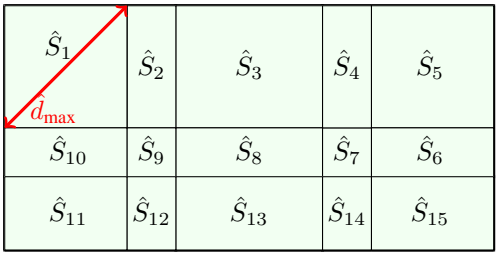
\end{example}

\subsection{Runs Consistency} \label{sec: runs consistency}

Due to the fact that the dynamics of the system have couplings between the agents, it is necessary to define timed runs that can be performed from each individual agent. Even though we have the individual WTS of each agent, the runs that the later generates may not be performed by an agent due to the constrained motion that is imposed by the coupling terms. Hence, we need to provide a tool that synchronizes the agents at each time step $\delta t$ and is able to determine which of the generated runs of the individual WTS can be performed by the agent (i.e., they are consistent runs). In order to address the aforementioned issue, we provide a centralized product WTS which captures the behavior of the coupled multi-agent system as a team, and the generated product run (see Def. \ref{def: consistent_runs}) can later be projected onto consistent individual runs. The following two definitions deal with the product WTS and consistent runs respectively.

\begin{definition} \label{def:product_TS}
	The product WTS $\mathcal{T}_p = (S_p, S_p^{\text{init}}, \longrightarrow_p)$ is defined as follows:
	\begin{itemize}
		\item $S_p = \hat{\mathbb{I}}^N$.
		\item $(s_1, \ldots, s_N) \in S^{\text{init}}$ if $s_i \in S_i^{\text{init}}$ for all $i \in \mathcal{I}$.
		\item $(\bf{l}, \bf{l}') \in \longrightarrow_p$ iff $l_i' \in \text{Post}_i(l_i, \text{pr}_i(\bf{l})), \forall \ i \in \mathcal{I}, \forall \ \bf{l} = (l_1, \ldots, l_N), \bf{l}' = (l'_1, \ldots, l'_N)$.
		\item $d_p: \longrightarrow_p \rightarrow \mathbb{T}$ : As in the individual WTS's case, the transition weight is $d_p(\cdot) = \delta t$.
	\end{itemize}
\end{definition}

The action labels, the atomic propositions and the labeling function in $\mathcal T_p$ are insignificant. Hence, without loss of generality, there were omitted from the tuple.

\begin{definition} \label{def: consistent_runs}
	Given the timed run
	\begin{align}
		&r_p^t = ((r_p^1(0), \ldots , r_p^N(0)),\tau_p(0))((r_p^1(1),\ldots,r_p^N(1)),\tau_p(1))\ldots \notag
	\end{align}
	of the WTS $\mathcal T_p$, the induced set of projected runs $$\{r_i^t = (r_p^i(0), \tau_p(0))(r_p^i(1), \tau_p(1)) \ldots  : i \in \mathcal I\}$$ of the WTSs $\mathcal T_1,\ldots, \mathcal T_N$, respectively will be called \textit{consistent runs}. Since the duration of each agent's transition is $\delta t$ it holds that $\tau_p(j) = j \cdot \delta t, j \geq 0$.
\end{definition}

Therefore, through the product WTS $\mathcal{T}_p$, we can always generate individual consistent runs for each agent. It remains to provide a systematic approach of how to determine consistent runs $\widetilde{r}_1, \ldots, \widetilde{r}_N$  which are associated with the corresponding service time words $\widetilde{w}_1^t, \ldots, \widetilde{w}_N^t$ (note that with tilde we denote the outcome of our solution approach) and, according to the Lemma \ref{lemma:compliant_WTS_runs_with_trajectories}, their corresponding compliant trajectories $x_1(t), \ldots, x_N(t)$ will satisfy the corresponding MITL formulas $\varphi_1, \ldots, \varphi_N$. and they are the solution to Problem \ref{problem: basic_prob}.

\begin{example} \label{ex: example_2}
	
An example that explains the notation that has been introduced until now is the following: Consider an agent (Fig. \ref{fig: example_1}) moving in the workspace with $\mathcal{N}(i) = \{1,2\}$, $S = \{S_\ell\}_{\ell \in \mathbb{I} = \{1, \ldots, 6\}}$ is the given cell decomposition from Problem \ref{problem: basic_prob}, $\bar{S} = \{\bar{S}_{l}\}_{l \in \mathbb{I} = \{1, \ldots, 28\}}$ is the cell decomposition that is the outcome of the abstraction and atomic propositions $\{p_1,\ldots,p_6\} = \{\rm orange, $ $green, blue$ $yellow, red, grey\}$. The red arrows represent both the transitions of the agent $i$ and its neighbors. The dashed lines indicate the edges in the network graph. For the atomic propositions we have that $L_i(14) = \{p_1\}, L_i(17) = \{p_5\},  L_i(10) = \{p_2\}, L_i(20) = \{p_4\}, L_{j_1}(28) = \{p_6\} = L_{j_1}(27), L_{j_1}(24) = \{p_5\}, L_{j_1}(22) = \{p_4\}, L_{j_2}(2) = \{p_1\}, L_{j_2}(12) = \{p_2\} = L_{j_2}(5), L_{j_2}(9) = \{p_3\}$. Note also the diameter of the cells $\hat{d}_{\text{max}} = d_{\text{max}}$. For the cell configurations we have:
	\begin{align}
		&Init:
		\begin{cases}
			{\bf{l}_i} = (14, 28, 2) \\
			{\bf{l}_{j_1}} = (28, 14) \\
			{\bf{l}_{j_2}} = (2, 14) \\
		\end{cases}
		Step 1:
		\begin{cases}
			{\bf{l}_i} = (17, 27, 13) \\
			{\bf{l}_{j_1}} = (27, 17) \\
			{\bf{l}_{j_2}} = (13, 17) \\
		\end{cases} \notag \\
		&Step 2:
		\begin{cases}
			{\bf{l}_i} = (10, 24, 5) \\
			{\bf{l}_{j_1}} = (24, 10) \\
			{\bf{l}_{j_2}} = (5, 10) \\
		\end{cases}
		Step 3:
		\begin{cases}
			{\bf{l}_i} = (20, 22, 9) \\
			{\bf{l}_{j_1}} = (22, 20) \\
			{\bf{l}_{j_2}} = (9, 20) \\
		\end{cases} \notag
	\end{align}
	which are actions to the corresponding transitions. The consistent timed runs are given as
	\begin{align}
		&r_i^t = (r_i(0) = 14, \tau_i(0) = 0) (r_i(1) = 17, \tau_i(1) = \delta t) \notag \\
		&(r_i(2) = 10, \tau_i(2) = 2 \delta t) (r_i(3) = 20, \tau_i(3) = 3 \delta t) \notag \\
		&r_{j_1}^t = (r_{j_1}(0) = 28, \tau_{j_1}(0) = 0) (r_{j_1}(1) = 27, \tau_{j_1}(1) = \delta t) \notag \\
		&(r_{j_1}(2) = 24, \tau_{j_1}(2) = 2 \delta t) (r_{j_1}(3) = 22, \tau_{j_1}(3) = 3\delta t) \notag \\
		&r_{j_2}^t = (r_{j_2}(0) = 2, \tau_{j_2}(0) = 0) (r_{j_2}(1) = 13, \tau_{j_2}(1) = \delta t) \notag \\
		&(r_{j_2}(2) = 5, \tau_{j_2}(2) = 2 \delta t) (r_{j_2}(3) = 9, \tau_{j_2}(3) = 3\delta t). \notag
	\end{align}
	It  can be observed that $r_i^t \models (\varphi_i = \Diamond_{[0, 6]}\{yellow\})$ if $3 \delta t \in [0, 6]$, $r_{j_1}^t \models (\varphi_{j1} = \Diamond_{[3, 10]}\{red\})$ if $2 \delta t \in [3, 10]$ and $r_{j_2}^t \models (\varphi_{j2} = \Diamond_{[3, 9]}\{blue\})$ if $3 \delta t \in [3, 9]$. For $\delta t = 1$, all the agents satisfy their goals.
	
	\begin{figure}[ht!]
		\centering
		\begin{tikzpicture}[scale = 0.80]
		
		\filldraw[fill=yellow!40, line width=.04cm] (0, 0) rectangle +(3, 3.0);
		\filldraw[fill=orange!40, line width=.04cm] (-7.5, -3) rectangle +(3, 3.0);
		\filldraw[fill=red!20, line width=.04cm] (-4.5, 0) rectangle +(4.5, 3.0);
		\filldraw[fill=black!10, line width=.04cm] (-7.5, 0) rectangle +(3, 3.0);
		\filldraw[fill=blue!20, line width=.04cm] (0, -3) rectangle +(3, 3.0);
		\filldraw[fill=green!40, line width=.04cm] (-4.5, -3) rectangle +(4.5, 3);
		
		\draw[step=1.5, line width=.04cm] (-7.5, -3) grid (3,3);
		

		\draw (-6.8,-0.6) node[circle, inner sep=0.8pt, fill=black, label={below:{$i$}}] (A) {};
		\draw (-3.8, 0.7) node[circle, inner sep=0.8pt, fill=black, label={below:{$ $}}] (B) {};
		\draw (-0.8, -0.7) node[circle, inner sep=0.8pt, fill=black, label={below:{$ $}}] (C) {};
		\draw (0.8, 0.7) node[circle, inner sep=0.8pt, fill=black, label={below:{$ $}}] (D) {};
		\draw (-5.3, -2.1) node[circle, inner sep=0.8pt, fill=black, label={below:{$j_2$}}] (E) {};
		\draw (-6.8, 2.3) node[circle, inner sep=0.8pt, fill=black, label={below:{$ $}}] (F) {};
		\draw (-5.0, 2.2) node[circle, inner sep=0.8pt, fill=black, label={below:{$ $}}] (G) {};
		\draw (-0.8, 2.2) node[circle, inner sep=0.8pt, fill=black, label={below:{$ $}}] (H) {};
		\draw (2.2, 2.2) node[circle, inner sep=0.8pt, fill=black, label={below:{$ $}}] (I) {};
		\draw (-3.8, -0.85) node[circle, inner sep=0.8pt, fill=black, label={below:{$ $}}] (K) {};
		\draw (-0.8, -2.30) node[circle, inner sep=0.8pt, fill=black, label={below:{$ $}}] (L) {};
		\draw (0.75, -0.8) node[circle, inner sep=0.8pt, fill=black, label={below:{$ $}}] (M) {};
		
		
		\node at (-6.8, 2.69) {$j_1$};

		
		\draw[->, red, line width=.04cm] (A) to [bend left=35] (B);
		\draw[->, red, line width=.04cm] (B) to [bend left=-35] (C);
		\draw[->, red, line width=.04cm] (C) to [bend left=35] (D);
		\draw[->, red, line width=.04cm] (F) to [bend left=15] (G);
		\draw[->, red, line width=.04cm] (G) to [bend left=15] (H);
		\draw[->, red, line width=.04cm] (H) to [bend right=15] (I);
		\draw[->, red, line width=.04cm] (E) to [bend right=15] (K);
		\draw[->, red, line width=.04cm] (K) to [bend right=15] (L);
		\draw[->, red, line width=.04cm] (L) to [bend right=15] (M);
		
		
		\draw[color = black,very thick,<->] (1.5,-3) -- (3,-1.5);
		\coordinate [label=below right:$\bf{d_{\max}}$] (N) at (2.2,-2);
		
		\node at (1.8, 2.7) {$22$};
		\node at (1.9, -0.3) {$8$};
		\node at (-7.2, 1.20) {$15$};
		\node at (-7.2, -1.9) {$1$};
		
		\node [red!40] at (-2.3, 3.5) {$S_5$};
		\node [yellow!100] at (1.7, 3.4) {$S_4$};
		\node [blue!20] at (1.7, -3.7) {$S_3$};
		\node [green!70] at (-2.3, -3.7) {$S_2$};
		\node [orange!60] at (-6.3, -3.7) {$S_1$};
		\node [black!30] at (-6, 3.5) {$S_6$};
		
		
		\node [red] at (-5.6, 0.95) {$\delta t$};
		\node [red] at (-2.5, -1.15) {$\delta t$};
		\node [red] at (-0.5, 0.5) {$\delta t$};
		
		\draw [dashed, black] (A) -- (E);
		\draw [dashed, black] (A) -- (F);
		\draw [dashed, black] (B) -- (G);
		\draw [dashed, black] (C) -- (H);
		\draw [dashed, black] (D) -- (I);
		\draw [dashed, black] (K) -- (B);
		\draw [dashed, black] (L) -- (C);
		\draw [dashed, black] (M) -- (D);
		\end{tikzpicture}
		\caption{Timed runs of the agents $i, j_1, j_2$}
		\label{fig: example_1}
	\end{figure}
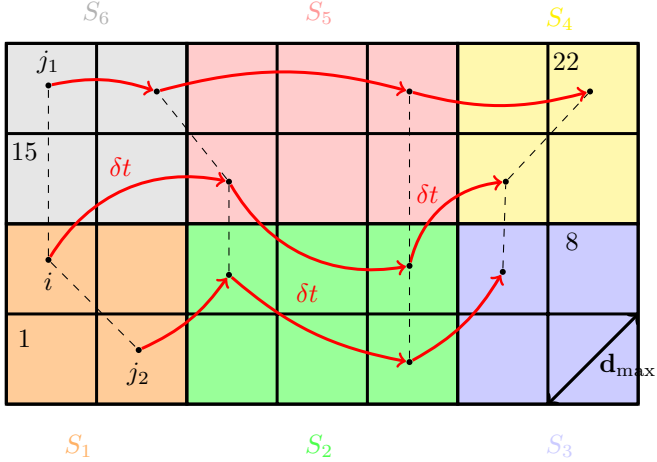
\end{example}

\begin{remark}
We chose to utilize decentralized abstractions, to generate the individual WTSs $\mathcal{I}, i \in \mathcal{I}$ for each agent and to compute the synchronized-centralized product WTS $\mathcal{T}_p$ for the following reasons: 
\begin{enumerate}
\item The state space of the centralized system to be abstracted is $\mathcal{X}^N \subseteq \mathbb{R}^{Nn}$, which is i) harder to visualize and handle ii) not naturally related to the individual specifications. Thus, it is more “natural” to define the specifications through the individual transition system of each agent corresponding to a discretization of $\mathcal{X}$ and then form the product to obtain the possible consistent satisfying plans.
\item Additionally, many centralized abstraction frameworks are based on approximations of the system's reachable set from a given cell
over the transition time interval. These require in the general nonlinear case  global dynamics properties and may avoid taking into account the finer dynamics properties of the individual entities, which can lead to more conservative estimates for large scale systems. 
\end{enumerate}
\end{remark}

\subsection{Controller Synthesis} \label{sec: synthesis}

The proposed controller synthesis procedure is described with the following steps:

\begin{enumerate}
	\item $N$ TBAs $\mathcal{A}_i, \ i \in \mathcal{I}$ that accept all the timed runs satisfying the corresponding specification formulas $\varphi_i, i \in \mathcal{I}$ are constructed.
	\item A B\"uchi WTS $\tilde{\mathcal T}_i = \mathcal{T}_i \otimes \mathcal{A}_i$ (see Def. \ref{def: buchi_WTS} below) for every $i \in \mathcal{I}$ is constructed. The accepting runs of $\tilde{\mathcal T}_i$ are the individual runs of the $\mathcal{T}_i$ that satisfy the corresponding MITL formula $\varphi_i, \ i \in \mathcal{I}$.
    \item We pick a set of accepting runs $\{\widetilde{r}^t_1, \ldots, \widetilde{r}^t_N\}$ from Step~2. We check if they are consistent according to Def. \ref{def: consistent_runs}. If this is true then we proceed with Step 5. If this is not true then we repeat Step 3 with a different set of accepting runs. At worst case, the number of repetitions that should be performed is finite; if a consistent set of accepting runs is not found, we proceed with the less efficient centralized, yet complete, procedure in Step 4.
	\item We create the product $\widetilde{\mathcal{T}}_p = \mathcal{T}_p \otimes \mathcal{A}_p$ where $\mathcal{A}_p$ is the TBA that accepts all the words that satisfy the formula $\varphi = \varphi_1 \wedge \ldots \wedge \varphi_N$. An accepting run $\widetilde{r}_p$ of the product is projected into the accepting runs $\{\widetilde{r}_1, \ldots, \widetilde{r}_N\}$. If there is no accepting run found in  $\mathcal{T}_p \otimes \mathcal{A}_p$, then Problem \ref{problem: basic_prob} has no solution.	
	\item The abstraction procedure allows to find an explicit feedback law for each
	transition in $\mathcal T_i$. Therefore,
	an accepting run $\widetilde{r}^t_i$ in $\mathcal T_i$ that takes the form of a sequence of transitions is realized in the system in \eqref{eq: system} via the corresponding sequence of feedback laws.
\end{enumerate}

In order to construct the Buchi WTSs $\widetilde{\mathcal{T}}_p$ and $\widetilde{\mathcal{T}}_i, i \in \mathcal{I}$ that were presented in Steps 2 and 4, we consider the following definition:
\begin{definition} \label{def: buchi_WTS}
	Given a WTS $\mathcal{T}_i =(S_i, S_{i}^{\text{init}}, Act_i, \longrightarrow_i, d_i, AP_i, \hat L_i)$, and a TBA $\A_i = (Q_i,  Q^\text{init}_i, C_i, \\ Inv_i, E_i, F_i, AP_i, \mathcal{L}_i)$ with $|C_i|$ clocks and let $C^{\mathit{max}}_i$ be the largest constant appearing in $\A_i$. Then, we define their \textit{B\"uchi WTS} $\widetilde{\T}_i = \mathcal{T}_i \otimes \A_i = (\widetilde{S}_i, \widetilde{S}_{i}^{\init}, \widetilde{Act}_i, {\rightsquigarrow}_{i}, \widetilde{d}_i, \widetilde{F}_i, AP_i, \widetilde{L}_i)$ as follows:
	\begin{itemize}
		\item {$\widetilde{S}_i \subseteq \{(s_i, q_i) \in S_i \times Q_i : \hat{L}_i(s_i) = \mathcal{L}_i(q_i)\} \times \mathbb{T}_\infty^{|C_i|} $.}
		\item $\widetilde{S}_{i}^{\init} = S_i^{\init} \times Q_i^{\init} \times \underbrace{\{0\} \times \ldots \times \{0\}}_{|C_i| \ products}$.
		\item $\widetilde{Act}_i = Act_i$.
		\item $(\widetilde{q}, {\bf{I}_i}, \widetilde{q} ') \in {\rightsquigarrow}_i$ iff
		\begin{itemize}
			\item[$\circ$] $\widetilde{q} = (s, q, \nu_1, \ldots, \nu_{|C_i|}) \in \widetilde{S}_i$, \\ $\widetilde{q} ' = (s', q', \nu_1', \ldots, \nu_{|C_i|}') \in \widetilde{S}_i$,
			\item[$\circ$] ${\bf{I}_i} \in Act_i$,
			\item[$\circ$] $(s, {\bf{I}_i}, s') \in \longrightarrow_i$, and
			\item[$\circ$] there exists $\gamma, R$, such that $(q, \gamma, R, q') \in E_i$, $\nu_1,\ldots,\nu_{|C_i|} \models \gamma$, $\nu_1',\ldots,\nu_{|C_i|}' \models Inv_i(q')$, and for all $i \in \{1,\ldots, |C_i|\}$
			\begin{equation*}
			\nu_i' =
			\begin{cases}
			0,      & \text{if } c_i \in R \\
			\nu_i + d_i(s, s'), &  \text{if }  c_i \not \in R \text{ and } \\ & \nu_i + d_i(s, s') \leq C^{\mathit{max}}_i \\			\infty, & \text{otherwise}.
			\end{cases}
			\end{equation*}
		\end{itemize}
		Then, $\widetilde{d}_i(\widetilde{q}, \widetilde{q}') = d_i(s, s')$.
		\item $\widetilde{F}_i = \{(s_i, q_i,\nu_1,\ldots,\nu_{|C_i|}) \in Q_i : q_i \in F_i\}$.
		\item $\widetilde{L}_i(s_i, q_i, \nu_1, \ldots, \nu_{|C_i|}) = \hat{L}_i(s_i)$.
	\end{itemize}
\end{definition}

The Buchi WTS $\widetilde{\mathcal{T}}_p$ is constructed in a similar way. Each B\"uchi WTS $\widetilde \T_i, i \in \mathcal I$ is in fact a WTS with a B\"uchi acceptance condition $\widetilde{F}_i$. A timed run of $\widetilde \T_i$ can be written as $\widetilde{r}_i^t = (q_i(0), \tau_i(0))(q_i(1), \tau_i(1)) \ldots$ using the terminology of Def. \ref{run_of_WTS}. It is \textit{accepting} if $q_i(i) \in \widetilde F_i$ for infinitely many $i \geq 0$.
An accepting timed run of  $\widetilde{\T}_i$ projects onto a timed run of $\T_i$ that satisfies the local specification formula $\varphi_i$ by construction. Formally, the following lemma, whose proof follows directly from the construction and and the principles of automata-based LTL model checking (see, e.g., \cite{katoen}), holds:

\begin{lemma} \label{eq: lemma_1}
	Consider an accepting timed run $\widetilde{r}_i^t = (q_k(0), \tau_i(0))(q_i(1), \tau_i(1)) \ldots$ of the B\"uchi WTS $\widetilde \T_k$ defined above, where $q_i(k) = (r_i(k), s_i(k), \nu_{i, 1}, \ldots, \nu_{i, M_i})$ denotes a state of $\mathcal{\widetilde T}_i$, for all $k \geq 1$. 
	The timed run $\widetilde{r}_i^t$ projects onto the timed run $r_i^t = (r_i(0), \tau_i(0))(r_i(1), \tau_i(1)) \ldots$ of the WTS $\mathcal{T}_i$ that produces the timed word $w(r_i^t) = (L_i(r_i(0)), \tau_i(0))(L_i(r_i(1)), \tau_i(1)) \ldots$ accepted by the TBA $\mathcal{A}_i$ via its run $\rho_i = s_i(0)s_i(1) \ldots$. Vice versa, if there exists a timed run $r_k^t = (r_k(0),\tau_k(0))(r_k(1),\tau_k(1))\ldots$ of the WTS $\T_k$ that produces a timed word $w(r_k^t) = (L_k(r_k(0)), \tau_i(0))(L_i(r_i(1)), \tau_i(1)) \ldots$ accepted by the TBA $\A_i$ via its run $\rho_i = s_i(0)s_i(1)\ldots$ then there exist the accepting timed run $\widetilde{r}_i^t = (q_i(0), \tau_i(0))(q_i(1), \tau_i(1)) \ldots$ of $\widetilde{\T}_i$, where $q_i(i)$ denotes $(r_i(k),s_i(k),\nu_{i,1}(i), \ldots, \nu_{i,M_i}(k))$ in $\widetilde{\T}_i$.
\end{lemma}

\begin{proposition}
	By following the procedure described in Sec. \ref{sec: synthesis} a sequence of controllers $v_1, \ldots, v_N$ can be designed (if there is a solution according to Steps 1-5) that guarantees the satisfaction of the formulas $\varphi_1, \ldots, \varphi_N$ of the agents $1, \ldots, N$ respectively, governed by dynamics as in \eqref{eq: system}. 
\end{proposition}

\subsection{Complexity} \label{sec:complexity}

Our proposed framework can handle all the expressivity of the MITL formulas according to the semantics of Definition \ref{def:mitl_semantics}. Denote by $|\varphi|$ the length of an MITL formula $\varphi$. A TBA $\mathcal{A}_i, i \in \mathcal{I}$ can be constructed in space and time $2^{\mathcal{O}(|\varphi_i)|}, i \in \mathcal{I}$. So by denoting with $\varphi_{\text{max}} = \text{max} \{ |\varphi_i\}, i \in \mathcal{I}$ the MITL formula with the longest length we have that the complexity of Step 1 is $2^{\mathcal{O}(|\varphi_{\text{max}})|}$. The model checking of Step 2 costs $\mathcal{O}(|\mathcal{T}_i| \cdot 2^{|\varphi_i|}), i \in \mathcal{I}$ where $|\mathcal{T}_i|$ is the length of the WTS $\mathcal{T}_i$ i.e., the number of its states. Thus, $\mathcal{O}(|\mathcal{T}_i| \cdot 2^{|\varphi_i|}) = \mathcal{O}(|S_i| \cdot 2^{|\varphi_i|}) = \mathcal{O}(|\hat{\mathbb{I}}| \cdot 2^{|\varphi_i|})$. The worst case of Step 2 costs $\mathcal{O}(|\mathcal{T}_{\text{max}}| \cdot 2^{|\varphi_{\text{max}}|})$ where $|\mathcal{T}_{\text{max}}|$ is the number of the states of the WTS which corresponds to the longest formula $\varphi_{\text{max}}$. Due to the fact that all the WTSs in Step 2 have the same number of states, it holds that the worst case complexity of Step 2 costs $\mathcal{O}(|\hat{\mathbb{I}}| \cdot 2^{|\varphi_{\text{max}}|})$. By denoting with $R_{iter}$ the finite number of repetitions of Step 3, we have the best case complexity as $\mathcal{O}(R_{\text{iter}} \cdot |\hat{\mathbb{I}}| \cdot 2^{|\varphi_{\text{max}}|})$, since the Step 3 is more efficient than Step 4. The worst case complexity of our proposed framework is when Step 4 is followed, which is $\mathcal{O}(|\hat{\mathbb{I}}|^N \cdot 2^{|\varphi_{\text{max}}|})$ where $|\hat{\mathbb{I}}|$ is the number of cells of the cell decomposition $\hat{S}$.

\begin{figure}[h]
	\centering
	\includegraphics[scale = 0.60]{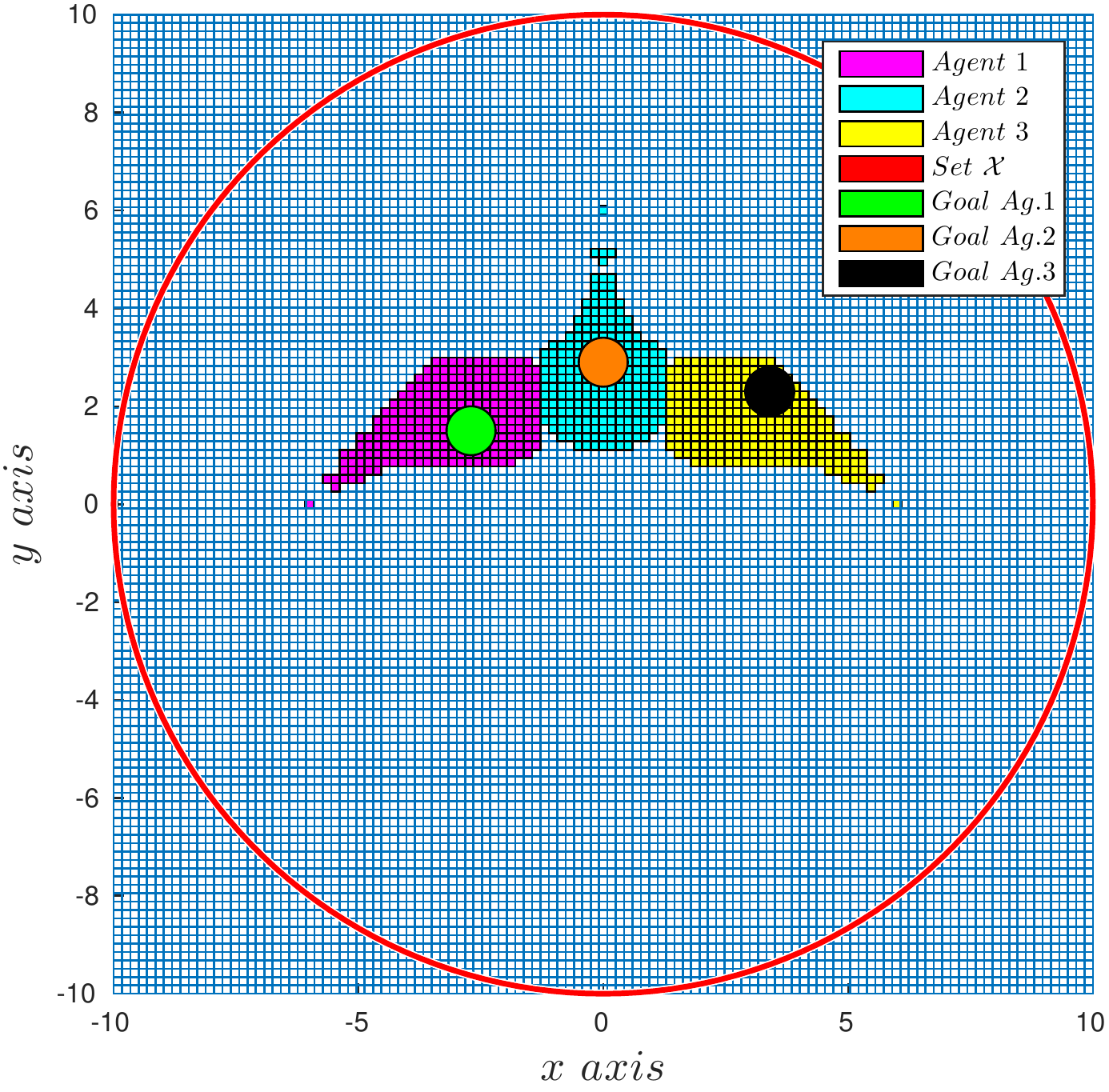}
	\caption{Space discretization, goal regions and reachable sets for each agent in a time horizon of $11 \delta t$ steps}
	\label{fig:simulation}
\end{figure}

\section{Simulation Results} \label{sec: simulation_results}

For a simulation example, a system of three agents with $x_i \in \mathbb{R}^2, \ i \in \mathcal{I} = \{1, 2, 3\}, \mathcal{E} = \{(1,2), (2,3)\}, \mathcal{N}(1) = \{2\} = \mathcal{N}(3), \mathcal{N}(2) = \{1, 3\}$ is considered. Their dynamics are given as $\dot{x}_1 = x_2-x_1+v_1, \dot{x}_2 = x_1+x_3-2x_2+v_2$ and $\dot{x}_3 = x_2-x_3+v_3$. The simulation parameters are set to $\bar{R} = 10, M = 20, v_{\text{max}} = 10, L_1 = \sqrt{2}, L_2 = 2, \delta t = 0.2$. The time step $\delta t$ is chosen during the abstraction process according to the formulas \eqref{dmax}, \eqref{deltat} and it is not chosen with reference to satisfaction of the MITL formulas. The workspace $[-10,10] \times [-10, 10] \subseteq \mathbb{R}^2$ is partitioned into cells and the initial agents' positions are set to $(-6,0), (0,6)$ and $(6,0)$ respectively. The specification formulas are set to $\varphi_1 = \Diamond_{[0.5, 1.7]} \{\rm green\}, \varphi_2 = \Diamond_{[1.0, 1.4]} \{\rm orange\}, \varphi_3 = \Diamond_{[0.7, 1.8]} \{black\}$ respectively and their corresponding TBAs are given in Fig. \ref{fig:TBA_example}.  The abstraction presented in this paper, the reachable cells of each agent as well as the goal regions are depicted in Fig. \ref{fig:simulation}. It can be observed that not all the individual runs satisfy the desired specification. By applying the five-step controller synthesis procedure that was presented in Sec. \ref{sec: solution}, the individual run of each agent satisfy the formulas $\varphi_1, \varphi_2$ and $\varphi_3$ in $6 \delta t$, $6 \delta t$ and $5 \delta t$ respectively. The simulation is performed in a horizon of $11 \delta t$ steps (as the steps that explained in the Example \ref{ex: example_2}). The product WTS has $45 \times 10^4$ states. The simulations were carried out in MATLAB Environment on a desktop with 8 cores, 3.60GHz CPU and 16GB of RAM.

\section{Conclusions and Future Work} \label{sec: conclusions}

A systematic method of both abstraction and controller synthesis of dynamically coupled multi-agent path-planning has been proposed, in which timed constraints of fulfilling a high-level specification are imposed to the system. The solution involves initially a boundedness analysis and secondly the abstraction of each agent's motion into WTSs and automata construction. The simulation example demonstrates our solution approach. Future work includes further computational improvement of the abstraction method and more complicated high-level tasks being imposed to the agents in order to exploit the expressiveness of MITL formulas.






\bibliography{references}
\bibliographystyle{ieeetr}
\addtolength{\textheight}{-12cm}

\end{document}